\documentclass{CSML}

\def\dOi{12(1:5)2016}
\lmcsheading%
{\dOi}
{1--30}
{}
{}
{Jun.~\phantom04, 2014}
{Mar.~\phantom09, 2016}
{}

\ACMCCS{[{\bf Theory of computation}]: Formal languages and
  automata theory---Regular languages}

\usepackage{enumitem}
\setlist[enumerate,1]{label=\hbox{$(\arabic*)$}}
\usepackage{amsmath}
\usepackage{marvosym}
\usepackage{microtype}
\usepackage[utf8x]{inputenc}
\usepackage[T1]{fontenc}
\usepackage{amssymb,amsthm,xspace,stmaryrd,mathrsfs,pdfsync,graphicx,verbatim}
\usepackage{tikz}
\usepackage[bookmarks=true,bookmarksopen=true,bookmarksopenlevel=1]{hyperref}
\usepackage{tabularx}

\usetikzlibrary{trees,arrows,patterns,shapes,snakes,fit,shadows,calc,automata}

\newcounter{sauvegarde}
\DeclareMathAlphabet\EuScript{U}{eus}{m}{n}
\SetMathAlphabet\EuScript{bold}{U}{eus}{b}{n}

\definecolor{my1}{cmyk}{0,.6,0,0}
\definecolor{my2}{cmyk}{.3,.0,.0,.0}

\newcommand\Ss{\ensuremath{\EuScript{S}}\xspace}
\newcommand\Gs{\ensuremath{\EuScript{G}}\xspace}
\newcommand\Ts{\ensuremath{\EuScript{T}}\xspace}
\newcommand\Rs{\ensuremath{\EuScript{R}}\xspace}
\newcommand\Hs{\ensuremath{\EuScript{H}}\xspace}
\newcommand\Cs{\ensuremath{\EuScript{C}}\xspace}

\newcommand\Kb{\ensuremath{\mathbf{K}}\xspace}
\newcommand\Lb{\ensuremath{\mathbf{L}}\xspace}
\newcommand\Hb{\ensuremath{\mathbf{H}}\xspace}
\newcommand\Fb{\ensuremath{\mathbf{F}}\xspace}

\newcommand{\efgame}{Ehrenfeucht-Fra\"iss\'e\xspace}
\newcommand\nat{\ensuremath{\mathbb{N}}\xspace}

\newcommand\Is{\ensuremath{\mathcal{I}}}

\DeclareMathOperator{\Sat}{Sat}
\DeclareMathOperator{\downclos}{\downarrow}

\newcommand{\siw}{\ensuremath{\Sigma_{2}(<)}\xspace}

\newcommand{\siu}{\ensuremath{\Sigma_{1}(<)}\xspace}

\newcommand{\sict}{\ensuremath{\Sigma_{3}}\xspace}

\newcommand{\ltl}{\ensuremath{\textup{LTL}}\xspace}

\newcommand\fodw{\ensuremath{\textup{FO}^2(<)}\xspace}

\newcommand{\fo}{\ensuremath{\text{FO}}\xspace}
\newcommand{\fow}{\ensuremath{\textup{FO}(<)}\xspace}

\newcommand\foeq[1]{\ensuremath{\equiv_{#1}}\xspace}
\newcommand\kfoeq{\foeq{k}}
\newcommand{\foclos}[2]{\ensuremath{[#1]_{\foeq{k}}}\xspace}

\newcommand\Sep{\ensuremath{\mathsf{Sep}}\xspace}

\newcommand\iword{\ensuremath{\omega}-word\xspace}
\newcommand\ilang{\ensuremath{\omega}-language\xspace}
\newcommand\iwords{\ensuremath{\omega}-words\xspace}
\newcommand\ilangs{\ensuremath{\omega}-languages\xspace}
\newcommand\sisemi{sub-\ensuremath{\omega}-semigroup\xspace}

\newcommand\isemi{\ensuremath{\omega}-semigroup\xspace}
\newcommand\isemis{\ensuremath{\omega}-semigroups\xspace}

\newcommand\Hrel{\ensuremath{{\mathscr{H}}}}

\newcommand{\unclos}[1]{\ensuremath{\llfloor{#1}\rrfloor}\xspace}

\theoremstyle{plain}
\newtheorem{theorem}[thm]{Theorem}
\newtheorem{corollary}[thm]{Corollary}
\newtheorem{proposition}[thm]{Proposition}
\newtheorem{lemma}[thm]{Lemma}
\newtheorem{fct}[thm]{Fact}
\newtheorem{remark}[thm]{Remark}
\newtheorem*{claim}{Claim}

\usepackage[english]{babel}

\let\leq\leqslant

\let\geq\geqslant
\let\emptyset\varnothing

\begin{document}

\title[SEPARATING REGULAR LANGUAGES WITH FIRST-ORDER LOGIC]
      {Separating Regular Languages with First-Order Logic}

\author[T.~Place]{Thomas~Place}
\author[M.~Zeitoun]{Marc~Zeitoun}

\address{LaBRI, Bordeaux University, France}
\email{\{thomas.place, marc.zeitoun\}@labri.fr} \thanks{Supported by ANR 2010 BLAN 0202 01
  FREC.\@ This study has also been carried out with financial support from the French
  State, managed by the French National Research Agency (ANR) in the frame of
  the ``Investments for the future'' Programme IdEx Bordeaux -- CPU
  (ANR-10-IDEX-03-02).}


\begin{abstract}
  Given two languages, a separator is a third language that contains the first
  one and is disjoint from the second one. We investigate the following
  decision problem, called \emph{separation}: given two regular
  languages of finite words, decide whether there exists a first-order
  definable separator. A more general problem was solved in an algebraic
  framework by Henckell in 1988, although the connection with separation was
  pointed out only in 1996, by Almeida. The result was then generalized
  by Henckell, Steinberg and Rhodes in 2010. In this paper, we present a new,
  self-contained and elementary proof of it, which actually covers the
  original result of Henckell.

  We prove that in order to answer this question, sufficient information can
  be extracted from semigroups recognizing the input languages, using a
  fixpoint computation, similar to that originally proposed by Henckell. Given
  as input a morphism recognizing both languages to be separated, this yields
  an \textsc{Exptime} algorithm for checking first-order separability.
  Moreover, the correctness proof of this algorithm yields a stronger result,
  namely a description of a possible separator. More precisely, one can
  compute a bound on the quantifier rank of potential separators, as well as a
  first-order formula that describes a separator, if there exists one.
  Finally, we prove that this technique can be generalized to answer the same
  question for regular languages of infinite words.
\end{abstract}

\maketitle

\section{Introduction}
\label{sec:intro}
In this paper, we investigate a decision problem on word languages:
the \emph{separation problem}. The problem is parametrized by a class
\Sep of \emph{separator languages} and is as follows: given as input
two regular word languages, decide whether there exists a third
language in \Sep containing the first language while being disjoint
from the second one.

More than the decision procedure itself, the primary motivation for
investigating this type of problem is the insight it gives on the class
\Sep. Indeed, the separation problem is a generalization of the
\emph{membership problem}, which is often considered as the right
approach to understand the expressive power of a class of languages.
In this restricted problem, one only needs to decide whether a single
input regular language already belongs to the class $\Sep$ under
investigation. Intuitively, in order to get such a decision procedure,
one has to consider \emph{all} regular languages simultaneously, which
requires a strong understanding of the \emph{expressive power} of
\Sep. Since regular languages are closed under complement, testing
membership can be achieved by testing whether the input is separable
from its complement. Therefore, membership can be reduced to
separation, which makes separation more general.

It turns out that separation is actually strictly more general than
membership and solving it requires a deeper understanding of the class
\Sep. More than the expressive power, it requires an understanding of
the \emph{discriminating power} of \Sep. This means that while
intrinsically more difficult, solving the separation problem is also
more rewarding than solving the membership problem. In both cases, the
problem amounts to finding a language in \Sep. However, in the
membership case, there is only one candidate, which is already known:
the input. Therefore, we start with a fixed recognizing device for
this unique candidate and powerful tools are available,
\emph{viz}.~the syntactic semigroup of the language, which is now
accepted as the natural tool for solving the membership problem for
word languages. In the separation case, there can be infinitely
many candidates as separators, which means that there is no fixed
recognition device that we can use. An even harder question then is to
actually construct a separator language in~\Sep.

Investigating the deeper separation problem can also be relevant when a pure
membership approach fails. Many natural classes of languages are built on top
of weaker classes. For example, in logic, more powerful classes can be built
on top of weaker ones by adding predicates to the signature. When
investigating the membership problem, a natural approach would be to first
obtain a solution for the weaker class and then to transfer it to the extended
one. However, this approach fails in general. Actually, many extensions of
classes of languages are known not to preserve decidability of
membership~\cite{ABR:Undec-Identity:92,Rhodes99-undec,DBLP:journals/ijac/Auinger10}.
The reason is that such a transfer result requires more information on the original
class than what a solution to membership provides. This makes the deeper
separation problem a more promising setting as already noted
in~\cite{AZ97-J,Steinberg:delay2001}. A recent example is the quantifier
alternation hierarchy of first-order logic: in~\cite{PZ:icalp14}, it was
proved that solving the separation problem for level $i$ in this hierarchy
yields the solution for the membership problem at level~($i+1$).

\smallskip\noindent{\bf First-order logic.} In this paper, we choose
\Sep as the class of \emph{languages definable by first-order
  sentences} (\emph{i.e.}, sets of words that satisfy some first-order
sentence).  In this context, the separation problem can be rephrased
as follows: given two regular languages as input, decide whether there
exists a first-order sentence that is satisfied by all words of the
first language, and by no word of the second one. Thus, such a
formula witnesses that the input languages are disjoint.

Within monadic second order logic, which defines on finite words
all regular languages, first-order logic is often considered as the
yardstick.  It is a robust class having several
characterizations~\cite{Diekert&Gastin:First-order-definable-languages:2008:a}. It
corresponds to star-free languages, and has the same expressive power
as linear temporal logic~\cite{kltl}.  In particular, it was the first
natural class for which the membership problem was proved to be
decidable.  This result, known as Schützenberger's
theorem~\cite{sfo,mnpfo}, served as a template and a starting point of
a line of research that successfully solved the membership problem for
a wide assortment of classes of regular languages. This makes
first-order logic the natural candidate to serve as the example for
devising a general approach to the separation~problem.

Schützenberger's theorem states that first-order definable languages
are exactly those whose syntactic semigroup is aperiodic, \emph{i.e.},
has only trivial subgroups. Since the syntactic semigroup of a
language is computable and aperiodicity is a decidable property, this
yields a decision procedure for membership. Schützenberger's original
proof has been refined over the years. Our own proof for separation by
first-order logic actually generalizes a more recent proof by
Wilke~\cite{wfo}. Similar results~\cite{tfo,pfo} make it possible to decide
first-order definability for languages of infinite words, or finite or
infinite Mazurkiewicz
traces. See~\cite{Diekert&Gastin:First-order-definable-languages:2008:a}
for a survey.

\smallskip\noindent
{\bf Contributions and main ideas.} We obtain our separation algorithm
for first-order logic by relying on a specific framework. A key idea
is that, in order to separate two regular languages with first-order
logic, one needs to consider more languages than just these two. One
has to consider a single morphism from $A^+$ into a finite semigroup
$S$ that recognizes them both and solve the separation problem simultaneously for all
pairs of languages that are recognized by it. Indeed,
the set of all languages recognized by a semigroup morphism has
structure: considering them all as a whole allows us to exploit
this structure.

More precisely, our framework is designed to reduce separation to the
following more general problem. Given a morphism $\alpha$ from $A^+$
into a finite semigroup $S$, we want to construct an \fo-partition of
$A^+$ (a {\bf finite} partition of $A^+$ into first-order definable
languages) that is an  ``optimal approximation'' of the languages
recognized by $\alpha$. The main point is that a necessary condition
for an \fo-partition to be optimal (for $\alpha$) is that any two
recognized languages are \fo-separable \emph{if and only if} they can
be separated by a language built as a union of languages in the
\fo-partition (however this condition is not sufficient, which is why
this problem is more general).

Our solution is presented as follows. First, we obtain a fixpoint
algorithm that, given a morphism $\alpha$ as input, computes an
object that we call the \emph{optimal imprint with respect to \fo on
  $\alpha$}. Intuitively, this ``optimal imprint'' contains
information about the \fo-partitions that are optimal for $\alpha$.
In fact, this information includes which pairs of languages recognized
by $\alpha$ are \fo-separable. In other words, this yields a decision
procedure for the separation problem associated to \fo. This fixpoint
algorithm is complemented by a generic technique for constructing
optimal \fo-partitions by induction (this is actually a byproduct of
the correctness proof of the fixpoint algorithm). This is of particular
interest as this yields an inductive way to build first-order
separators when they exist.

An important observation is that the ``\emph{optimal imprint with
respect to \fo on $\alpha$}'' that our fixpoint algorithm computes
is actually an alternate definition of a previously known notion:
the so-called \emph{aperiodic pointlike sets} (whose original definition is algebraic and very different from the
one we use in this paper). While we never use this fact in the paper,
it connects our results to those of
Henckell~\cite{Henckell:Pointlike-sets:-finest-aperiodic:1988:a} (see
also~\cite{DBLP:journals/ijac/HenckellRS10a,qt} which answers the problem for even more general classes). Indeed,
in~\cite{Henckell:Pointlike-sets:-finest-aperiodic:1988:a}, Henckell does not consider the separation problem: his main objective is to
find an algorithm that computes these aperiodic pointlike sets.
In fact, the connection between the separation
problem and the pointlike sets was only observed later
by Almeida~\cite{MR1709911}\footnote{This connection is the analogue of the equivalence
  $\ref{item:cov2sep1}\Longleftrightarrow\ref{item:cov2sep2}$ in our
  Theorem~\ref{thm:seppart}.}. Hence, our fixpoint algorithm and its
proof can be viewed as a new proof of Henckell's result: one can
compute the aperiodic pointlike sets of a semigroup.

Note however that our approach is vastly different from that of
Henckell. In particular, it is more rewarding with respect to the
separation problem. Indeed, the motivations and the proofs
of~\cite{Henckell:Pointlike-sets:-finest-aperiodic:1988:a,DBLP:journals/ijac/HenckellRS10a}
are purely algebraic and provide no intuition on the underlying logic.
Our contributions differ from those of~\cite{Henckell:Pointlike-sets:-finest-aperiodic:1988:a,DBLP:journals/ijac/HenckellRS10a}
in several~ways.
\begin{itemize}[leftmargin=*]
\item First, we give a new and self-contained proof that the separation
  problem by first-order languages is decidable. It is independent from those
  of~\cite{Henckell:Pointlike-sets:-finest-aperiodic:1988:a,DBLP:journals/ijac/HenckellRS10a},
  and relies on elementary ideas and notions from language theory only, making
  it accessible to computer scientists. We do not use any involved
  construction from semigroup theory: we work directly with the logic
  itself.   As mentioned above, the proof refines Wilke's membership algorithm~\cite{wfo}.

\item Second, when the input languages are separable, our approach makes it
  possible to inductively compute a first-order formula that defines a
  separator: we have a generic way to construct optimal \fo-partitions.
  In addition, we provide a bound on the \emph{expected quantifier rank} of a
  potential separator.

\item Third, as a consequence of our algorithm, we obtain~an \textsc{Exptime}
  upper bound (while complexity is not investigated
  in~\cite{Henckell:Pointlike-sets:-finest-aperiodic:1988:a}, a rough analysis
  yields an \textsc{Expspace} upper~bound).

\item Finally, the techniques
  of~\cite{Henckell:Pointlike-sets:-finest-aperiodic:1988:a,DBLP:journals/ijac/HenckellRS10a}
  are tailored to work with finite words only. We also solve the
  separation problem for languages of \emph{infinite words} by
  first-order definable languages, by a smooth extension of our
  techniques.
\end{itemize}
Since we do not follow the proofs
of~\cite{Henckell:Pointlike-sets:-finest-aperiodic:1988:a,DBLP:journals/ijac/HenckellRS10a},
it is not surprising that we obtain a different algorithm. However, we are
able to derive two variations of it, which allows us to give an alternate
and elementary correctness proof of Henckell's original algorithms.

\medskip\noindent{\bf Related work.} First-order logic has a number of
important fragments. The separation question makes sense when choosing such
natural subclasses as classes of separators. It has already been solved for
the case of local fragments, such as locally testable (LT) and locally
threshold testable languages (LTT), although the problem is already NP-hard
starting from two DFAs as input, while membership is known to be
polynomial~\cite{Beauquier&Pin:Languages-scanners:1991:a}. The algebraic
varieties associated to the classes LT and LTT in Eilenberg's correspondence
are well-known, namely the class \textsf{LSl} of all finite local
semilattices, and the semidirect product $\mathsf{Acom}*\mathsf{D}$ of
commutative and aperiodic semigroups with right zero semigroups. Using these
correspondences and the algebraic interpretation of separation given
in~\cite{MR1709911}, algebraic proofs were given, both for
LT~\cite{Costa&Nogueira:Complete-reducibility-pseudovariety:2009:a,Costa:Free-profinite-locally-idempotent:2001:a}
and for
LTT, via~\cite{Beauquier&Pin:Languages-scanners:1991:a,Straubing:Finite-semigroup-varieties-form:1985:a,Steinberg:98,Steinberg:01}.
Although indirect, these proofs actually provide more information than what is
needed for separation alone. A direct and elementary approach for both classes
was also presented in~\cite{ltltt:2013,PvRZ:LTT:14}.

The separation problem is also decidable for the fragment of first-order logic made of Boolean
combinations of \siu sentences (that is, first order sentences without any
quantifier alternation), as a consequence of~\cite{MR1709911,AZ97-J}, and then
obtained directly and independently in~\cite{sep_icalp13,PvRZ:mfcs}. It has then been shown to be decidable for
the first fragments of first-order logic in the quantifier alternation
hierarchy, namely the ones consisting of \siw\cite{PZ:icalp14}, respectively
of $\sict(<)$ sentences~\cite{pseps3} (\emph{i.e}, first order sentences
of the form $\exists^*\forall^*\varphi$, respectively of the form
$\exists^*\forall^*\exists^*\varphi$, with $\varphi$ quantifier-free). In view
of the aforementioned transfer result, this yields decidability
of membership for the next level, $\Sigma_{4}(<)$. Within this hierarchy,
membership remains open for level 5 and above (hence, separation is open
for level 4 and~above).

Finally, the problem has also been
investigated for the fragment \fodw of first-order logic using 2 variables
only, and again has been proven to be decidable~\cite{PvRZ:mfcs}.

\medskip\noindent {\bf Paper outline.}  We first give the necessary
definitions and terminology: languages and semigroups for finite words are
defined in Section~\ref{sec:prelims} and first-order logic is defined in
Section~\ref{sec:fo}. Section~\ref{sec:main} is devoted to the presentation of
our algorithm solving first-order separation through the computation of sets
that cannot be distinguished by first-order logic. Sections~\ref{sec:correc}
and~\ref{sec:comp} are devoted to proving the soundness and completeness of
this algorithm, respectively. In Section~\ref{sec:altalgo}, we present
alternate versions of our algorithm. In Section~\ref{sec:omega}, we recall the
preliminary definitions for tackling the separation problem in the setting of
infinite words. In Section~\ref{sec:omegasep}, we state a
generalization to infinite words of our algorithm, for which we prove
soundness in Section~\ref{sec:corr-algor} and completeness in
Section~\ref{sec:compl-algor}.

\medskip
This paper is the journal version of~\cite{PZ:lics14}. From the
conference version, the missing proofs have been added, separation is now
presented in a generic, language-theoretic setting, and the proof
of the algorithm has been entirely rewritten so that it now
constructs an actual separator by induction when it exists.

\section{Preliminaries}
\label{sec:prelims}
\newcommand\one{\textup{1}}

In this section, we provide terminology for words, semigroups and
languages. All the definitions are for finite words. We
delay the definitions for infinite words to Section~\ref{sec:omega}.

\medskip
\noindent {\textbf{Semigroups.}} A semigroup is a set $S$ equipped
with an associative operation $s\cdot t$ (often written $st$). A
monoid is a semigroup $S$ having an identity element $1_S$,
\emph{i.e.}, such that $s\cdot1_S=1_S\cdot s=s$ for all $s\in S$.
Finally, a group is a monoid such that every element $s$ has an
inverse $s^{-1}$, \emph{i.e.}, such that $s\cdot s^{-1}=s^{-1}\cdot
s=1_S$.

Given a \emph{finite} semigroup $S$, it is folklore and easy to see that
there is an integer $\omega(S)$ (denoted by $\omega$ when $S$ is
understood) such that for all $s$ in $S$,
$s^\omega$ is idempotent:~$s^\omega=s^\omega s^\omega$.

\medskip
\noindent
{\textbf{Words, Languages, Morphisms.}} We fix a finite alphabet
$A$. We denote by $A^+$ the set of all nonempty finite
words and by $A^{*}$ the set of all finite words over $A$. If $u,v$ are words, we
denote by $u \cdot v$ or by $uv$ the word obtained by the concatenation of $u$ and
$v$. Observe that $A^+$ (resp.\ $A^{*}$) equipped with the
concatenation operation is a semigroup (resp.\ a monoid).

For convenience, we only consider languages that do not contain the
empty word. That is, a language is a subset of $A^+$ (this does not
affect the generality of the argument). We work with
regular languages, \emph{i.e.}, languages definable by
\emph{nondeterministic finite automata}~(NFA).

We shall exclusively work with the algebraic representation of regular
languages in terms of semigroups.  We say that a language \emph{$L$ is
  recognized by a semigroup $S$} if there exists a semigroup
morphism $\alpha : A^+ \rightarrow S$ and a subset $F \subseteq S$
such that $L = \alpha^{-1}(F)$. It is well known that a language is
regular if and only if it can be recognized by a
\emph{finite}~semigroup. Moreover, from any NFA recognizing some
language $L$, one can compute a canonical semigroup recognizing $L$,
called the \emph{syntactic semigroup} of $L$.

When working on separation, we consider as input two regular languages
$L_0,L_1$. It will be convenient to have a single semigroup
recognizing both of them, rather than having to deal with two
objects. Let $S_0,S_1$ be semigroups recognizing $L_0,L_1$ together
with the associated morphisms $\alpha_0, \alpha_1$, respectively. Then, $S_0 \times
S_1$ equipped with the componentwise multiplication $(s_0,s_1) \cdot
(t_0,t_1)=(s_0 t_0,s_1 t_1)$ is a semigroup that recognizes both $L_0$ and $L_1$
with the morphism $\alpha : w \mapsto (\alpha_0(w),\alpha_1(w))$.
From now on, we work with such a single semigroup recognizing both
languages, and we call $\alpha$ the associated morphism.

\medskip
\noindent {\bf Semigroup of Subsets.} As explained in the
introduction, our separation algorithm works by computing special
subsets of a semigroup recognizing both input languages. Intuitively,
these subsets are those that cannot be distinguished by first-order
logic. More precisely, by \emph{special subset}, we mean that any first-order
definable language has an image under $\alpha$ that either contains
\emph{all} elements of the subset, or \emph{none} of them. For this
reason, we work with the semigroup of~subsets.
\smallskip
Let $S$ be a semigroup. Observe that the set $2^S$ of subsets of
$S$ equipped with the operation $$T \cdot T' = \{s \cdot s' \mid s\in
T, \quad s' \in T'\}$$ is a semigroup, that we call the
\emph{semigroup of subsets of $S$}. Note that $S$ can be viewed as a
subsemigroup of $2^S$, since $S$ is isomorphic to the semigroup
$\big\{\{s\} \mid s \in S\big\} \subseteq 2^{S}$.
We denote by $\Rs,\Ss,\Ts,\dots$ subsemigroups of a
semigroup of subsets.

\smallskip\noindent\textbf{Downset $\downclos\Ss$.} Let $\Ss\subseteq
2^S$ be any subset of $2^{S}$. We define the \emph{downset} $\downclos\Ss$ of \Ss as
\[
\downclos \Ss=\{T\in2^S\mid\exists T'\in\Ss,\quad T\subseteq T'\}.
\]
Clearly, we have $\Ss\subseteq\downclos\Ss$. Moreover, if $\Ss$ is a
subsemigroup of $2^S$, it is easy to check that $\downclos\Ss$ is a
subsemigroup as well.

\smallskip\noindent\textbf{Union \unclos{\Ss}.} For $\Ss \subseteq
2^S$ any subset of $2^S$, we define $\unclos{\Ss} \subseteq S$, the
\emph{union} of $\Ss$, as the set
\[
\unclos{\Ss} = \bigcup_{T \in \Ss} T \subseteq S
\]
We call \emph{index} of \Ss the size of its union, \emph{i.e.},
$|\unclos{\Ss}|$. By definition, we have the following
fact.

\begin{fct} \label{fct:uprod}
Set $\Ss \subseteq 2^S$ and $\Ts \subseteq 2^S$, then $\unclos{\Ss
  \cdot \Ts} = \unclos{\Ss} \cdot \unclos{\Ts}$.
\end{fct}

\section{First-Order Logic and Separation}
\label{sec:fo}
This section is devoted to the definition of first-order logic on
words. See~\cite{Thomas:Languages-automata-logic:1997:a,Diekert&Gastin:First-order-definable-languages:2008:a,bookstraub} for details on
this classical notion.

\medskip
\noindent
{\textbf{First-Order Logic.}} We view words  as logical structures
composed of a sequence of positions labeled over $A$. We denote by $<$
the linear order over the positions. We work with first-order logic
\fow using a unary predicate $a(x)$ for each $a \in A$,
which selects
positions $x$ labeled with an $a$, as well as a binary predicate for the
linear order $<$. A language $L$ is said to be \emph{first-order definable} if there
exists an \fow formula $\varphi$ such that $L = \{w \in A^+ \mid w
\models \varphi\}$. We write \fo for the class of all first-order definable
languages.

There are many known characterizations of the class of first-order
definable languages. Kamp's Theorem~\cite{kltl} states that it is
exactly the class of languages definable in linear temporal logic
\ltl.  In~\cite{mnpfo}, it was proven that this is also the class of
languages that can be recognized with \emph{counter-free automata} as
well as the class of star-free languages (\emph{i.e.}, languages
definable by a regular expression that may use complement, but does
not use the Kleene star). This result bridged the gap with
Schützenberger's Theorem~\cite{sfo}, which characterizes star-free
languages as those whose syntactic semigroup is \emph{aperiodic} (\emph{i.e.},
all its elements $s$ satisfy the equality $s^\omega = s^{\omega+1}$).
The separation algorithm that we present in this paper can be viewed
as a generalization of Schützenberger's Theorem. In particular, we
reprove this theorem as a simple corollary of our algorithm.
Note that conversely, using Sch\"utzenberger's result as a black box doesn't seem
to help much to obtain a simpler proof with our approach.

\medskip\noindent{\textbf{Separation}.} Given languages $L,L_0,L_1$,
we say that $L$ \emph{separates} $L_0$ from $L_1$ if
\begin{equation*}
  L_0 \subseteq L \text{ and } L_1 \cap L = \emptyset.
\end{equation*}
Given a class of languages \Cs, the pair $(L_0,L_1)$ is said to be
\emph{\Cs-separable} if some language $L \in \Cs$ separates $L_0$
from~$L_1$. Note that when \Cs is closed under complementation (for
example when $\Cs = \fow$), $(L_0,L_1)$ is \Cs-separable if and only
if $(L_1,L_0)$ is. Therefore, we simply say that $L_0$ and $L_1$ are
\Cs-separable in this case.

In this paper, we present an algorithm that decides whether two
regular languages are \fo-separable. Let us give an example of two
languages that are not \fo-separable.

\begin{exa} \label{ex:main}
Let $K_0=(aa)^*$, $K_1=(aa)^*a$ and
\begin{align*}
    L_0 &= (bK_0bK_1)^+,\\
    L_1 &= (bK_0bK_1)^*bK_0.
\end{align*}

It is well known that $a^{2^k}$ and $a^{2^k-1}$ cannot be
distinguished by any \fo-sentence of quantifier rank~$k$, see
\emph{e.g.}~\cite{bookstraub} (recall here that the \emph{quantifier rank} of a first-order formula $\varphi$ is the length of the largest sequence of nested
quantifiers in $\varphi$ --- the rank is a  usual way to classify first-order formulas).
Therefore, $K_0$ and $K_1$ are not
\fo-separable. Reusing this argument then shows that $L_0$ and $L_1$
are not \fo-separable either. We shall explain below how this is
detected by our algorithm.
\end{exa}

The main result of the paper is the following theorem.

\begin{theorem} \label{th:main}
Let $L_0,L_1$ be two regular languages recognized by a morphism
$\alpha: A^+ \to S$ into a finite semigroup. The two following items
hold.
\begin{enumerate}
\item\label{item:thmain:1} One can decide in {\sc Exptime} with respect to $|S|$ whether
  $L_0$ and $L_1$ are \fo-separable.
\item\label{item:thmain:2} When $L_0$ and $L_1$ are \fo-separable, one can construct an
  actual FO-separator defined by a formula of quantifier rank at most
  $|A|2^{|S|^2}$.
\end{enumerate}
\end{theorem}

\noindent The proof of Theorem~\ref{th:main} is postponed to
Sections~\ref{sec:main}, \ref{sec:correc} and~\ref{sec:comp}. In
Section~\ref{sec:main} we present our decision procedure and
Section~\ref{sec:correc} and~\ref{sec:comp} are devoted to proving
soundness and completeness of this procedure. Note that most of our
efforts are aimed to obtaining an algorithm for Item~1 of the theorem.
We actually obtain the second item as a byproduct of the completeness
proof of Section~\ref{sec:comp}: this proof is constructive and can be
used to build an actual separator by induction (which turns out to
have rank at most $|A|2^{|S|^2}$), when it exists.

\section{Separation Algorithm}
\label{sec:main}
In this section, we define a general framework which is tailored
to the investigation of the separation problem. We then use it to
obtain a separation algorithm in the special case when the class
of separators is given by first-order logic, \emph{i.e.}, to prove
Theorem~\ref{th:main}. An important remark is that the problem that
we actually consider and solve is slightly more general than
separation. In particular, this problem takes an arbitrary number of
languages as input rather than just two. Let us first explain our
motivation for considering such a generalization.

In the separation problem, we are given a single semigroup morphism
$\alpha$ that recognizes the two input languages $L_0,L_1$ that need
to be separated. However, in general, a single morphism recognizes
several different languages, not just these two. Moreover, while these
other languages are not the ones we aim at separating, $L_0$ and $L_1$
are built-up from them. This makes all these languages relevant when
working with $L_0$ and $L_1$. Therefore, our approach is to consider
them all simultaneously in a problem that generalizes separation:
\emph{computing an \fo-partition that is optimal} for the morphism
$\alpha$.

\medskip
We organize the section in three parts. First, we present our
framework in a general context (\emph{i.e.}, for an arbitrary class
of
separators \Cs) and connect it to the separation problem. In the
second part, we apply this framework to first-order logic and use
it to obtain a separation algorithm and to prove Theorem~\ref{th:main}.
Finally, in the third part, we illustrate this algorithm on
Example~\ref{ex:main}.

\subsection{Definition}

For the definitions, we assume that an arbitrary class of languages
\Cs over our fixed alphabet $A^+$ is fixed. Moreover, we need \Cs to satisfy the three following
properties:
\begin{enumerate}
\item \Cs is nonempty and closed under boolean operations.
\item \Cs is closed under right and left quotients: for any $w \in
  A^+$ and $L \in \Cs$, we have
  \[
    w^{-1}L \stackrel{\text{def}}= \{u\in A^+ \mid wu \in L\} \in \Cs \quad \text{and} \quad Lw^{-1} \stackrel{\text{def}}= \{u\in A^+ \mid uw \in L\} \in \Cs.
  \]
\item \Cs only consists of regular languages.
\end{enumerate}

It is straightforward to verify that \fo satisfies these three
properties. Note that the objects that we define below make sense even
when \Cs does not satisfy  these three properties. However, we will
need these properties to make the connection with the separation problem.

\medskip
\noindent
{\bf \Cs-Partitions and Imprints.} Assume that an alphabet $A$ is
fixed. A \emph{\Cs-partition} (of $A^+$) is a {\bf finite} partition
$\Kb = \{K_1,\dots,K_m\}$ of $A^+$ such that all languages $K_i$ in \Kb
belong to~\Cs. Note that since \Cs is non-empty and closed under
boolean operations, $A^+$ belongs to~$\Cs$. Therefore, there exists at least
one \Cs-partition, namely $\{A^+\}$.

When we have a morphism $\alpha: A^+ \rightarrow S$ and a
\Cs-partition \Kb in hand, our main interest will be to know how good
\Kb is at separating languages recognized by $\alpha$: what are the
languages recognized by $\alpha$ that can be separated by a union of
languages in \Kb ? This information is captured by a new object that
we associate to each \Cs-partition and each morphism, the
\emph{imprint of the partition on the morphism}.

Given a morphism $\alpha: A^+ \rightarrow S$ into a finite semigroup
$S$ and a \Cs-partition \Kb. The \emph{imprint of \Kb on $\alpha$} is
defined as the set
\[
  \Is[\alpha](\Kb) = \{T \in 2^S \mid \text{there exists $K \in \Kb$
    such that $T \subseteq \alpha(K)$}\}.
\]
In other words, $T \in \Is[\alpha](\Kb)$ if and only if there exists a
language in \Kb that intersects $\alpha^{-1}(t)$ for all $t \in T$.
Observe that by definition, an imprint on $\alpha$ is a subset of
$2^S$. Hence, since $S$ is finite, there are finitely many possible
imprints on $\alpha$. We present three simple properties of imprints. The first one states that an imprint always
contains some trivial elements.

\begin{fact}
  \label{fct:trivial}
  Let $\alpha: A^+ \rightarrow S$ be a morphism into a finite semigroup
  $S$ and \Kb be a \Cs-partition. Then, $\{\{\alpha(w)\} \mid w \in
  A^+\} \subseteq \Is[\alpha](\Kb)$.
\end{fact}

\begin{proof}
  For any $w \in A^+$, there exists $K \in \Kb$ such that $w \in K$ (\Kb
  is a partition of $A^+$). Hence, $\{\alpha(w)\} \subseteq \alpha(K)$
  and $\{\alpha(w)\} \in \Is[\alpha](\Kb)$.
\end{proof}

The second property states that any imprint is closed under downset.

\begin{fact}
  \label{fct:downclos}
  Let $\alpha: A^+ \rightarrow S$ be a morphism into a finite semigroup
  $S$ and \Kb be a \Cs-partition. Then $\Is[\alpha](\Kb) = \downclos
  \Is[\alpha](\Kb)$.
\end{fact}

\begin{proof}
  By definition, $\Is[\alpha](\Kb) \subseteq \downclos
  \Is[\alpha](\Kb)$. Let us prove the converse inclusion. Set $T \in
  \downclos \Is[\alpha](\Kb)$. By definition, there exists $T' \in
  \Is[\alpha](\Kb)$ such that $T \subseteq T'$. By definition of
  imprints, we obtain $K \in \Kb$ such that $T' \subseteq \alpha(K)$.
  Therefore, $T \subseteq T' \subseteq \alpha(K)$ and $T \in\Is[\alpha](\Kb)$. Note that we have shown that $\Is[\alpha](\Kb)=\downclos\alpha(\Kb)$.
\end{proof}

The third property connects imprints to the separation problem: the
imprint of \Kb on $\alpha$ records which languages recognized by
$\alpha$ can be separated with $\Kb$.

\begin{lemma} \label{lem:imprint}
  Let $\alpha: A^+ \rightarrow S$ be a morphism into a finite semigroup
  $S$ and \Kb be a \Cs-partition. Let $L_1,L_2$ be two languages
  recognized by $\alpha$ and let $T_1,T_2 \subseteq S$ be the
  corresponding accepting sets. The two following conditions are
  equivalent:
  \begin{enumerate}
  \item\label{ita:1} for all $t_1 \in T_1$ and all $t_2 \in T_2$, we have $\{t_1,t_2\}
    \not\in \Is[\alpha](\Kb)$.
  \item \label{ita:2} $L_1$ and $L_2$ can be separated by a union of languages in
    $\Kb$.
  \end{enumerate}
\end{lemma}

\begin{proof}
  Assume first that Item~\ref{ita:1} holds and set $K = \bigcup_{\{K'
    \in \Kb \mid K' \cap L_1 \neq \emptyset\}} K'$. Since $\Kb$ is a
  partition of $A^+$, we have $L_1 \subseteq K$ by definition. Moreover,
  we know from Item~\ref{ita:1} that no language $K' \in \Kb$
  intersects both $L_1$ and $L_2$. It follows that $K \cap L_2 =
  \emptyset$: $K$ separates $L_1$ from $L_2$ and Item~\ref{ita:2} holds.

  Assume now that Item~\ref{ita:2} holds. Since \Kb is a partition, this
  means that no language $K \in \Kb$ intersects both $L_1$ and $L_2$. It
  follows from the definition of imprints that for all $t_1 \in T_1$ and
  all $t_2 \in T_2$, we have $\{t_1,t_2\} \not\in \Is[\alpha](\Kb)$.
\end{proof}

An important remark is that, in general, the imprint of \Kb on
$\alpha$ contains more than just separation related information. For
example, assume that $S = \{s_1,s_2,s_3\}$ and consider two
\Cs-partitions \Kb and $\Kb'$ having the following imprints on
$\alpha$:
\[
  \begin{array}{lll}
    \Is[\alpha](\Kb) & = &
                           \{\emptyset,\{s_1\},\{s_2\},\{s_3\},\{s_1,s_2\},\{s_1,s_3\},\{s_2,s_3\},\{s_1,s_2,s_3\}\},
    \\
    \Is[\alpha](\Kb') & = &
                            \{\emptyset,\{s_1\},\{s_2\},\{s_3\},\{s_1,s_2\},\{s_1,s_3\},\{s_2,s_3\}\}.
  \end{array}
\]
From the separation point of view, we know from
Lemma~\ref{lem:imprint} that \Kb and $\Kb'$ are equally useless (they
cannot be used to separate any pair of nonempty languages recognized by
$\alpha$). However, we also know from the imprints that $\Kb'$ is
``better'' as it contains no language that intersects
$\alpha^{-1}(s_1),\alpha^{-1}(s_2)$ and $\alpha^{-1}(s_3)$ at the same
time.

In view of Lemma~\ref{lem:imprint}, the smaller the imprint on
$\alpha$ of a \Cs-partition is, the better this \Cs-partition is at
separating languages recognized by $\alpha$. We use this remark to
define the notion of \emph{optimal} \Cs-partition.

\medskip
\noindent
{\bf Optimal \Cs-Partitions.} Given a morphism $\alpha: A^+
\rightarrow S$ into a finite semigroup $S$ and a \Cs-partition \Kb,
we say that \Kb is \emph{optimal for $\alpha$} if for any
\Cs-partition $\Kb'$,
\[
  \Is[\alpha](\Kb) \subseteq \Is[\alpha](\Kb')
\]
We can use the fact that \Cs is closed under intersection to prove
that for any morphism $\alpha$, there always exists a \Cs-partition that is
optimal for $\alpha$.

\begin{lemma} \label{lem:optimal}
  Let $\alpha: A^+ \rightarrow S$ be a morphism into a finite semigroup
  $S$. Then there exists a \Cs-partition that is optimal for $\alpha$.
\end{lemma}

\begin{proof}
  We prove that for any two \Cs-partitions $\Kb'$ and $\Kb''$, there
  exists a third \Cs-partition $\Kb$ such that, $\Is[\alpha](\Kb)
  \subseteq \Is[\alpha](\Kb')$ and $\Is[\alpha](\Kb) \subseteq
  \Is[\alpha](\Kb'')$. Since there are only finitely possible imprints
  on $\alpha$, the lemma will follow.

  We set $\Kb = \{K' \cap K'' \mid K' \in \Kb' \text{ and } K'' \in
  \Kb''\}$. Since $\Kb'$ and $\Kb''$ were $\Cs$-partitions and \Cs is
  closed under intersection, \Kb remains a \Cs-partition. Finally, it is
  immediate from the definitions that  $\Is[\alpha](\Kb) \subseteq
  \Is[\alpha](\Kb')$ and $\Is[\alpha](\Kb) \subseteq \Is[\alpha](\Kb'')$.
\end{proof}

Observe that the proof of Lemma~\ref{lem:optimal} is non-constructive.
Given a morphism $\alpha$, computing an actual optimal \Cs-partition
for $\alpha$ is a difficult problem in general. In fact, as seen in
Theorem~\ref{thm:seppart} below, this is more general than solving
\Cs-separability for any pair of languages recognized by $\alpha$.
Before we present this theorem, let us make an important observation
about optimal \Cs-partitions.

By definition, given a morphism $\alpha$, all \Cs-partitions that are
optimal for $\alpha$ have the same imprint on $\alpha$. Hence, this
unique imprint is a canonical object for \Cs and $\alpha$. We call it
the \emph{optimal imprint with respect to\/ \Cs on\/ $\alpha$} and we
denote it by $\Is_\Cs[\alpha]$:
\[
  \Is_\Cs[\alpha] = \Is[\alpha](\Kb) \quad \text{for any optimal \Cs-partition \Kb
    of $\alpha$}.
\]
Note that, as an imprint, $\Is_\Cs[\alpha]$ satisfies
Fact~\ref{fct:trivial} and Fact~\ref{fct:downclos}: $\{\{\alpha(w)\} \mid w \in
A^+\} \subseteq \Is_\Cs[\alpha]$ and $\Is_\Cs[\alpha] = \downclos
\Is_\Cs[\alpha]$. Moreover, using our three hypotheses on \Cs (note
that this is where we need the second and third ones), one can prove
another convenient property: $\Is_\Cs[\alpha]$ is a subsemigroup of
$2^S$ (\emph{i.e.}, it is closed under multiplication).

\begin{lemma} \label{lem:usemi}
  Let $\alpha: A^+ \rightarrow S$ be a morphism into a finite semigroup
  $S$. Then $\Is_\Cs[\alpha]$ is a subsemigroup of $2^S$: for all $R,T \in
  \Is_\Cs[\alpha]$, $RT \in \Is_\Cs[\alpha]$.
\end{lemma}

\begin{proof}
  Let $R,T \in \Is_\Cs[\alpha]$ with $R = \{r_1,\dots,r_m\}$ and $T =
  \{t_1,\dots,t_n\}$. We prove that $RT \in \Is_\Cs[\alpha]$. Let \Kb be
  a \Cs-partition of $A^+$. We have to prove $RT \in \Is[\alpha](\Kb)$,
  \emph{i.e.}, that there exists $K \in \Kb$ such that $RT \subseteq \alpha(K)$.
  This is a consequence of the following claim.

  \begin{claim}
    There exist $u_1,\dots,u_m \in A^+$ and $v_1,\dots,v_n\in A^+$ such
    that $\alpha(u_i) = r_i$ for $i \leq m$ and $\alpha(v_j) = t_j$ for
    $j \leq n$ and,
    \begin{itemize}
    \item for any $w \in A^+$, there exists $K \in \Kb$ such that
      $u_1w,\dots,u_mw \in K$.
    \item for any $w \in A^+$, there exists $K \in \Kb$ such that
      $wv_1,\dots,wv_n \in K$.
    \end{itemize}
  \end{claim}
  Before we prove the claim, let us finish the proof of the lemma. Using
  the first item of the claim for $w = v_1$, we obtain a language $K \in
  \Kb$ such that $u_1v_1,u_2v_1,\dots,u_mv_1 \in K$. Similarly, for all
  $i \leq m$, we can use the second item of the claim for $w = u_i$ and
  we obtain a language $K_i \in \Kb$ such that $u_iv_1,u_iv_2,\dots,u_iv_n
  \in K_i$. Note that each language $K_i$ contains the word $u_iv_1$,
  which also belongs to $K$. Hence, since $\Kb$ is a partition of $A^+$,
  we have $K = K_1 = \cdots = K_m$ and $K$ contains the word $u_iv_j$
  for all $i \leq m$ and $j \leq n$. Since $\alpha(u_iv_j) = r_it_j$,
  this exactly says that $RT \subseteq \alpha(K)$, which terminates the
  proof of the lemma.

  \smallskip
  It now remains to prove the claim. We prove the existence of the words $v_1,
  \dots,v_n$. The proof for that of $u_1,\dots,u_m$ is
  symmetric. Observe that for any $w \in A^+$, the set $\Lb_{w} =
  \{w^{-1}K \mid K \in \Kb\}$ is a \Cs-partition of $A^+$ (recall that
  \Cs is assumed to be closed under quotients). Moreover, since \Cs
  contains only regular languages, all $K \in \Kb$ are regular
  languages and by Myhill-Nerode Theorem, they have finitely many left
  quotients. It follows that the set $\{\Lb_{w} \mid w \in A^+\}$ is
  finite. Hence using the fact that \Cs is closed under boolean
  operations, we can construct a new \Cs-partition $\Lb$ that refines
  all partitions $\Lb_{w}$ for $w \in A^+$.

  Since $\Lb$ is a \Cs-partition and $T \in \Is_\Cs[\alpha]$, we know
  that there exists $L \in \Lb$ such that $T \subseteq \alpha(L)$.
  This means that $L$ contains $n$ words $v_1,\dots,v_n \in
  A^+$ such that $\alpha(v_j) = t_j$ for $j \leq n$. We now prove that
  $v_1,\dots,v_n\in A^+$ satisfy the conditions of the claim. Set
  $w \in A^+$, we know that there exists $L' \in \Lb_w$ such that $L
  \subseteq L'$ (\Lb refines $\Lb_w$). This means that $v_1,\dots,v_n
  \in L'$. Finally, by definition, $L' = w^{-1}K$ for some $K \in \Kb$,
  hence $wv_1,\dots,wv_n \in K$.
\end{proof}

\noindent
{\bf From \Cs-Partitions to Separation.} We can now connect
\Cs-partitions and optimal imprints to the separation problem for
\Cs.

\begin{theorem} \label{thm:seppart}
  Let $\alpha: A^+ \to S$ be a morphism into a finite semigroup $S$. Let
  $L_1,L_2$ be two languages recognized by $\alpha$ and let $T_1,T_2
  \subseteq S$ be the corresponding accepting sets. The following
  properties are equivalent:
  \begin{enumerate}
  \item\label{item:cov2sep1} $L_1$ and $L_2$ are \Cs-separable.
  \item\label{item:cov2sep2} for all $t_1 \in T_1$ and all $t_2 \in T_2$, we have $\{t_1,t_2\}
    \not\in \Is_\Cs[\alpha]$.
  \item\label{item:cov2sep3} for any \Cs-partition \Kb that is optimal
    for $\alpha$, $L_1$ and $L_2$ are separable by a union of languages
    in $\Kb$.
  \end{enumerate}

\end{theorem}

\begin{proof}
  We prove that $\ref{item:cov2sep3} \Rightarrow \ref{item:cov2sep1}
  \Rightarrow \ref{item:cov2sep2} \Rightarrow \ref{item:cov2sep3}$.
  Let us first assume that $\ref{item:cov2sep3}$ holds, \emph{i.e.},
  that for any \Cs-partition \Kb that is optimal for $\alpha$, $L_1$
  and $L_2$ are separable by a union of languages in $\Kb$. Since
  there exists at least one \Cs-partition that is optimal for $\alpha$
  (see Lemma~\ref{lem:optimal}), $L_1$ can be separated from $L_2$
  with a union of languages in \Cs. Since \Cs is closed under union,
  this separator is in \Cs and $\ref{item:cov2sep1}$ holds.

  We now prove that $\ref{item:cov2sep1} \Rightarrow
  \ref{item:cov2sep2}$. Assume that $\ref{item:cov2sep1}$ holds,
  \emph{i.e.}, that $L_1$ is \Cs-separable from $L_2$. This means that
  there exists a language $K \in \Cs$ such that $L_1 \subseteq K$ and $K
  \cap L_2 = \emptyset$ (\emph{i.e.}, $L_2 \subseteq A^+ \setminus K$). Since
  \Cs is closed under complementation, $A^+ \setminus K \in \Cs$ and
  $\Kb = \{K,A^+ \setminus K\}$ is a \Cs-partition. By construction,
  for all $t_1 \in T_1$ and all $t_2 \in T_2$, $\{t_1,t_2\} \not\in
  \Is[\alpha](\Kb)$. Hence $\{t_1,t_2\} \not\in \Is_\Cs[\alpha]$ since
  $\Is_\Cs[\alpha] \subseteq \Is[\alpha](\Kb)$ by definition.

  It remains to prove that $\ref{item:cov2sep2} \Rightarrow
  \ref{item:cov2sep3}$. Assume that for all $t_1 \in T_1$ and all $t_2
  \in T_2$, $\{t_1,t_2\} \not\in \Is_\Cs[\alpha]$ and let $\Kb$ be an
  optimal \Cs-partition for $\alpha$. Since \Kb is optimal, we know
  from our hypothesis that for all $t_1 \in T_1$ and all $t_2
  \in T_2$, $\{t_1,t_2\} \not\in \Is[\alpha](\Kb)$. Hence it
  follows from Lemma~\ref{lem:imprint} that $L_1$ can be
  separated from $L_2$ by a union of languages in \Kb.
\end{proof}

In view of Theorem~\ref{thm:seppart}, given a class of languages \Cs
that satisfies the appropriate properties, a general approach to the
separation problem can be devised as follows.
\begin{enumerate}
\item Present an algorithm which takes a morphism $\alpha: A^+ \to
  S$ as input and computes $\Is_\Cs[\alpha]$. Thanks to
  Theorem~\ref{thm:seppart}, this allows us to decide whether any two
  languages recognized by the input morphism $\alpha$ are \Cs-separable.

  Typically, this
  algorithm should be a lowest fixpoint: $\Is_\Cs[\alpha]$ is computed
  as the smallest set $\Sat_\Cs(\alpha)$ which contains the singletons
  $\{\alpha(w)\}$ for $w \in \alpha(A^+)$ (see Fact~\ref{fct:trivial})
  and is closed under a set of rules that is specific to \Cs. Note
  that, in view of Fact~\ref{fct:downclos}, one of these rules should always be closure under downset $\downclos$, and in view of Lemma~\ref{lem:usemi}, one of these rules should
  always be closure under multiplication.

\item To prove that the algorithm is sound, \emph{i.e.}, that
  $\Sat_\Cs(\alpha) \subseteq \Is_\Cs[\alpha]$, one needs to prove that
  for any computed set $T$, the imprint of any \Cs-partition \Kb must
  contain $T$. This is usually simple and involves \efgame arguments.
\item To prove that the algorithm is complete, \emph{i.e.}, that
  $\Is_\Cs[\alpha] \subseteq \Sat_\Cs(\alpha)$, one needs to construct
  a \Cs-partition \Kb whose imprint on $\alpha$ is included in
  $\Sat_\Cs(\alpha)$. By definition, this proves that $\Is_\Cs[\alpha]
  \subseteq \Is[\alpha](\Kb) \subseteq \Sat_\Cs(\alpha)$, hence, this
  proves completeness. We actually get more from this construction:
  combining it with the knowledge that the algorithm is also correct,
  we obtain $\Is_\Cs[\alpha] = \Is[\alpha](\Kb)$: the
  \Cs-partition \Kb that we construct is actually optimal for
  $\alpha$. By Item~\ref{item:cov2sep3} in Theorem~\ref{thm:seppart}, we get
  a way to construct an actual separator in \Cs of two \Cs-separable languages recognized
  by $\alpha$.
\end{enumerate}
This terminates the presentation of the general approach. We now apply
it to the special case when the class \Cs is \fo.

\subsection{A Separation Algorithm for First-Order Logic}

Fix a morphism $\alpha: A^+ \to S$ into a finite semigroup $S$. We
describe a lowest fixpoint algorithm for computing the optimal imprint
with respect to \fo on $\alpha$: $\Is[\alpha]$. Note that from now on
we work with \fo only. Therefore, we simply write $\Is[\alpha]$ for
$\Is_\fo[\alpha]$.

Set $\Ss$ as a subsemigroup of $2^S$. We define $\Sat(\Ss)$, the
\emph{saturation of \Ss}, as the smallest subset of $2^S$ that contains
\Ss and is closed under the three following operations:
\begin{enumerate}
\item\label{eq:oper1} Downset: $\Sat(\Ss) = \downclos
  \Sat(\Ss)$.
\item\label{eq:oper2} Multiplication: for any $T,T'
  \in \Sat(\Ss)$, we have $TT' \in \Sat(\Ss)$.
\item\label{eq:oper}  \fo-Closure: for any $T \in \Sat(\Ss)$,
  $T^{\omega}\cup T^{\omega+1} \in \Sat(\Ss)$.
\end{enumerate}

Note that it is immediate that one can compute $\Sat(\Ss)$ from
$\Ss$ using a lowest fixpoint algorithm. Finally, we define
$\Sat(\alpha)$ as $\Sat(\Ss)$ for $\Ss = \{\{\alpha(w)\} \mid w \in
A^+\}$.

An interesting observation is that only Operation~\ref{eq:oper}
is specific to first-order logic in the definition of $\Sat$. Indeed,
we already know from the generic presentation that $\Is[\alpha]$
contains $\{\{\alpha(w)\} \mid w \in A^+\}$ and is closed under
downset and multiplication. This terminates the presentation of the
algorithm, we state its correctness in the following proposition.

\begin{proposition} \label{prop:algoworks} Set $\alpha: A^+ \to S$ as
  a morphism into a finite semigroup $S$. Then,
  \[
    \Is[\alpha] = \Sat(\alpha).
  \]
\end{proposition}

Since $\Sat(\alpha)$ is computable, Proposition~\ref{prop:algoworks}
immediately implies that so is $\Is[\alpha]$. Using
Theorem~\ref{thm:seppart}, this yields the decidability of the
separation problem for first-order logic. A simple analysis of the
lowest fixpoint procedure shows an {\sc Exptime} complexity upper
bound. This proves the first item of Theorem~\ref{th:main}, as we now
show.

\begin{proof}[Proof of the first item of Theorem~\ref{th:main}]
          By Theorem~\ref{thm:seppart}, it
  suffices to prove that one can compute $\Is[\alpha]$ in {\sc Exptime}
  in the size of $S$. Indeed, it then suffices to test whether there
  exists $T \in \Is[\alpha]$ such that $\alpha(L_1) \cap T \neq
  \emptyset$ and $\alpha(L_2) \cap T \neq \emptyset$. This can also be
  achieved in {\sc Exptime} by testing all possible candidates $T$. By
  Proposition~\ref{prop:algoworks}, we know that computing $\Is[\alpha]$
  can be done by computing $\Sat(\alpha)$.

  By definition, $\Sat(\alpha) \subseteq 2^S$. This means that the
  number of steps the algorithm needs to reach the fixpoint is at most
  exponential in~$S$. Therefore, it suffices to prove that each step can be
  done in {\sc Exptime} to conclude that the whole computation can also be
  done in {\sc Exptime}. Each step requires computing $T^{\omega} \cup
  T^{\omega+1}$ for at most $|2^S|$ subsets $T$. Each computation can be done
  in {\sc Exptime}, since $T^\omega$ is equal to some $T^m$ for $m\leq|2^S|$
  such that $T^m=T^{2m}$.   \end{proof}

We postpone the proof of Item~\ref{item:thmain:2} of Theorem~\ref{th:main} (the
bound on the quantifier rank of the separator) to
Section~\ref{sec:comp} where we prove the difficult direction of
Proposition~\ref{prop:algoworks}: $\Is[\alpha] \subseteq
\Sat(\alpha)$. As explained, the proof amounts to constructing an
optimal \fo-partition~for~$\alpha$.

\smallskip
Another interesting observation about our saturation algorithm is that
it can be viewed as a generalization of Schützenberger's
Theorem~\cite{sfo,mnpfo}. Indeed, a language is first-order
\emph{definable} if and only if its syntactic semigroup is aperiodic.
One definition of aperiodicity is that a semigroup is aperiodic if and
only if it satisfies the identity $s^{\omega} = s^{\omega+1}$. The
counterpart to this definition can be found in the main operation of
our saturation procedure, Operation~\ref{eq:oper} (which is the only
non-generic operation). This observation raises another question:
could Operation~\ref{eq:oper} be replaced to reflect alternate
definitions of aperiodicity while retaining
Proposition~\ref{prop:algoworks}? We shall see in
Section~\ref{sec:altalgo} that this is indeed~possible. Another
consequence of this observation is that we can reprove
Schützenberger's Theorem as a simple corollary of
Proposition~\ref{prop:algoworks}.

\begin{corollary}[Schützenberger's Theorem] \label{cor:carac}
  Let $L$ be a regular language. Then $L$ can be defined in \fo if and
  only if its syntactic semigroup is aperiodic.
\end{corollary}

\begin{proof}
  It is known that a language is definable in \fo if and only if all
  languages recognized by its syntactic semigroup are definable in
  \fo as well (this is actually not specific to \fo and true for all
  classes of languages that are ``Varieties'', see~\cite{Pin13:MPRI} for
  example). Hence, if $S$ is the syntactic semigroup of $L$ and $\alpha:
  A^+ \to S$ the associated (surjective) morphism, $L$ is definable in \fo if and only if
  $\{\alpha^{-1}(s) \mid s \in S\}$ is an \fo-partition. The imprint on
  $\alpha$ of this \fo-partition is $\{\{\alpha(w)\} \mid w \in A^+\}
  \cup \{\emptyset\}$, which is equal to $\{\{s\}\mid s\in S\}\cup\{\emptyset\}$, since
  $\alpha$ is surjective. Therefore, $L$ is definable in \fo if and only if
  $\Sat(\alpha) = \Is[\alpha] = \{\{s\}\mid s\in S\}\cup\{\emptyset\}$ (see Proposition~\ref{prop:algoworks}). By definition
  of $\Sat(\alpha)$, this is equivalent to $s^\omega = s^{\omega+1}$ for
  all $s \in S$.
\end{proof}

It now remains to prove Proposition~\ref{prop:algoworks}. In
Section~\ref{sec:correc}, we prove that $\Sat(\alpha) \subseteq
\Is[\alpha]$. This corresponds to soundness of the algorithm: all
computed sets indeed belong to $\Is[\alpha]$. Finally, in
Section~\ref{sec:comp}, we focus on the proof of the most difficult
direction, which is the second one: $\Is[\alpha]\subseteq
\Sat(\alpha)$. It implies completeness of the algorithm, that is, that
every set belonging to $\Is[\alpha]$ is actually computed by the algorithm.

\smallskip
We finish this section by running the algorithm, to show that it
detects that the languages of Example~\ref{ex:main} are not
\fo-separable.

\subsection{Example~\ref{ex:main}, continued.}
To start our algorithm, we first need a semigroup morphism recognizing
both $L_0$ and $L_1$. Observe that both languages are recognized by
the automaton below, with 4 as final state for $L_0$, and 2 as final state for~$L_1$.
Therefore, its transition semigroup~$S$ recognizes both
languages\footnote{Recall that the transition semigroup consists of
  all partial mappings induced by words from the state set to itself. It
  is easy to see that it recognizes the language accepted by the
  automaton, see~\cite[Sec.~3.1]{Pin13:MPRI}.}.
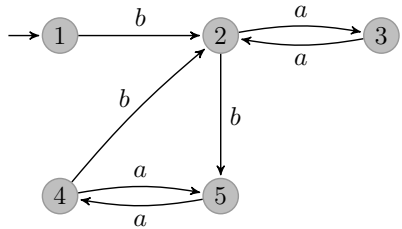
\begin{figure}[h]
  \begin{center}
    \scalebox{.89}{      \begin{tikzpicture}[node distance=2.4cm, initial text={}, ->,>=stealth',shorten >=1pt,auto, semithick]
        \tikzstyle{every
          state}=[circle,fill=black!25,draw=black!40,text=black,minimum
        size=15pt,inner sep=0pt]

        \node[initial,state] (A) {1};
        \node[state] (B) [right of=A] {2};
        \node[state] (C) [right of=B] {3};
        \node[state] (D) [below of=A, yshift=0mm] {4};
        \node[state] (E) [below of=B, yshift=0mm] {5};

        \path (A) edge node {$b$} (B)
        (B) edge [bend left = 10] node {$a$} (C) edge node {$b$} (E)
        (C) edge [bend left = 10] node {$a$} (B)
        (D) edge [bend left = 10] node {$a$} (E) edge [bend left = 5] node [yshift=-.75mm,xshift=.5mm] {$b$} (B)
        (E) edge [bend left = 10] node {$a$} (D);
      \end{tikzpicture}}
  \end{center}
  \caption{Automaton recognizing both $L_0$ and $L_1$}
  \label{fig:automtonA}
\end{figure}
The recognizing morphism $\alpha:A^+\to S$ maps a word to the partial
function it defines from states to states. We still denote the
images of $a,b\in A$ by $a,b\in S$, respectively.  It is easy
to see that $L_0=\alpha^{-1}(b^2a)$ and $L_1=\alpha^{-1}(\{b,b^2ab\})$.

We use Theorem~\ref{thm:seppart} to show that $L_0$ and $L_1$ are
not \fo-separable: we have to find $s_0 \in \alpha(L_0)$ and $s_1 \in
\alpha(L_1)$ such that $\{s_0,s_1\} \in \Is[\alpha]$. We claim that
$s_0 = b^2a$ and $s_1 = b^2ab$ satisfy this property. We actually show
that $\{s_0,s_1\}$ is computed as an element of $\Sat(\alpha)$, which,
by Proposition~\ref{prop:algoworks}, implies that it belongs to
$\Is[\alpha]$.

\smallskip

By definition, $\{a\},\{b\} \in \Sat(\alpha)$. Then, note that
$\{a\}^\omega=\{a^2\}$ and $\{a\}^{\omega+1}=\{a\}$. Therefore, by
definition of Operation~\ref{eq:oper}, we have
$\{a,a^2\}\in\Sat(\alpha)$.  Using Operation~\ref{eq:oper2}, we then
obtain that $X=\{a,aa\}\cdot\{b\}=\{ab,aab\} \in \Sat(\alpha)$.
Now, Operation~\ref{eq:oper} yields $Y= X^\omega\cup X^{\omega+1}\in
\Sat(\alpha)$. Computing $Y$ shows that $\{bab,bab^2\}\subseteq Y$.
Finally, using Operation~\ref{eq:oper2}, we obtain that $T=\{b\}\cdot
Y\cdot\{a,a^2\} \in \Sat(\alpha)$. One can then verify that
$\{b^2a,b^2ab\} \subseteq T$. Therefore, we get from
Operation~\ref{eq:oper1} that $\{b^2a,b^2ab\} \in \Sat(\alpha)$, as claimed.

\section{Soundness of the Algorithm}
\label{sec:correc}
In this section we prove soundness of our algorithm, that is the
inclusion $\Sat(\alpha) \subseteq \Is[\alpha]$ in
Proposition~\ref{prop:algoworks}. Recall that we work with a morphism
$\alpha: A^+ \rightarrow S$ into a finite semigroup~$S$. By definition of $\Sat(\alpha)$, we need to
prove that $\Is[\alpha]$ contains the set $\{\{\alpha(w)\} \mid
w \in A^+\}$ and is closed under Downset, Multiplication and
\fo-Closure.

We already know from Fact~\ref{fct:trivial} that $\Is[\alpha]$
contains the set $\{\{\alpha(w)\} \mid w \in A^+\}$. Furthermore,
closure under Downset and Multiplication follows from
Fact~\ref{fct:downclos} and Lemma~\ref{lem:usemi}. Therefore, we only
need to prove that $\Is[\alpha]$ is closed under \fo-Closure. This is
what we do now. To present the argument, we need an alternate
definition of $\Is[\alpha]$.

Given two words $w,w' \in A^+$ and $k \in \nat$, we write $w \foeq{k}
w'$ to denote the fact that $w$ and $w'$ satisfy the same \fo-formulas of
quantifier rank $k$. One can verify that for all $k$, \foeq{k} is an
equivalence relation of finite index and that each class can be
defined in \fo.

\begin{lemma} \label{lem:optequ}
Let $T \in 2^S$. Then $T \in \Is[\alpha]$ if and
only if for all $k \in \nat$, there exists an equivalence class $W$
of \foeq{k} such that $T \subseteq \alpha(W)$.
\end{lemma}

\begin{proof}
  Assume first that for all $k \in \nat$, there exists an equivalence
  class $W$ of \foeq{k} such that $T \subseteq \alpha(W)$. To prove that
  $T \in \Is[\alpha]$, we show that $T \in \Is[\alpha](\Kb)$ for any
  \fo-partition \Kb. By definition, there exists $k \in \nat$ such that all
  languages in \Kb can be defined by a formula of quantifier rank $k$ (recall
  that \fo-partitions are finite partitions). It follows that elements of \Kb
  are unions of classes of \foeq{k}, hence $W$ is included in some language $K$
  of $\Kb$. By hypothesis, $T \subseteq \alpha(W)$. It follows that
  $T \subseteq \alpha(W) \subseteq \alpha(K)$, which terminates the proof of
  this direction.

  Conversely, assume that $T \in \Is[\alpha]$ and set $k \in \nat$. Set \Kb as
  the partition of $A^+$ into equivalence classes of \foeq{k}. By definition,
  \Kb is an \fo-partition. Hence, there exists an equivalence class
  $W \in \Kb$ of \foeq{k} such that $T \subseteq \alpha(W)$.
\end{proof}

We can now prove that $\Is[\alpha]$ is closed under \fo-closure to conclude
the soundness proof: set
$T = \{t_1,\dots,t_n\} \in \Is[\alpha]$ and set $R = T^\omega \cup
T^{\omega+1}$, we need to prove that $R \in \Is[\alpha]$. We use
Lemma~\ref{lem:optequ}: for all $k \in \nat$, we shall find an equivalence class
$V$ of \foeq{k} such that $R \subseteq
\alpha(V)$.

Set $k \in \nat$. Using the other direction of Lemma~\ref{lem:optequ}
for $T$, we obtain an equivalence class $W$ of \foeq{k} such that $T
\subseteq \alpha(W)$. Set $V = W^{2^k\omega} \cup W^{2^k\omega+1}$.
By definition, $R \subseteq \alpha(V)$. Therefore, it suffices to
prove that $V$ is included in an equivalence class of $\kfoeq$
(\emph{i.e.}, that all words in $V$ are pairwise equivalent).
This
is a consequence of the following lemma, which is easy and folklore, and whose
simple proof is omitted here. It relies on \efgame games, details can be
found in~\cite{bookstraub}.

\begin{lemma} \label{lem:efconcat}
Set $k \in \nat$, then.
  \begin{enumerate}
  \item\label{item:5} For all $u_1,u_2,v_1,v_2 \in A^+$, if
    $u_1\foeq{k}v_1$ and $u_2\foeq{k}v_2$, then $u_1\cdot u_2 \foeq{k}
    v_1 \cdot v_2$.
  \item\label{item:6} For all $u\in A^+$ and all $k>0$, we have $u^{2^k}\foeq{k}u^{2^k+1}$.
  \end{enumerate}
\end{lemma}

Let us now conclude the proof. Pick some arbitrarily chosen word $w \in W$.  By
Lemma~\ref{lem:efconcat}~\ref{item:5}, it is immediate that any word of $V$
is $\kfoeq$-equivalent to either $u_0 = w^{2^k\omega} \in W^{2^k\omega}$ or to
$u_1 = w^{2^k\omega+1} \in W^{2^k\omega+1}$. To conclude that all words of
$V$ are $\kfoeq$-equivalent, it remains to prove that $u_0 \kfoeq u_1$,
which follows directly from Lemma~\ref{lem:efconcat}~\ref{item:6}.

\section{Completeness of the Algorithm}
\label{sec:comp}
In this section, we prove the most interesting inclusion from
Proposition~\ref{prop:algoworks}: $\Is[\alpha] \subseteq \Sat(\alpha)$. We use
induction to construct an \fo-partition \Kb whose imprint on $\alpha$ belongs
to $\Sat(\alpha)$. This proves that
$\Is[\alpha] \subseteq \Is[\alpha](\Kb) \subseteq \Sat(\alpha)$. For the rest
of this section, we assume fixed a morphism $\alpha: A^+ \rightarrow S$ into a
finite semigroup. We state the induction in the following proposition. Note
that, in order to set up the induction, we have to start from a morphism from
some free monoid $B^+$, where $B$ is an \emph{arbitrary} alphabet, into an \emph{arbitrary}
subsemigroup \Ss of $2^S$. This is because, following a proof from
Wilke~\cite{wfo}, we argue by induction on the size of (the index of) the
semigroup, and of the alphabet.

\begin{proposition} \label{prop:pumping}
Let $\Ss$ be a subsemigroup of $2^S$ and $\beta : B^+ \rightarrow
\Ss$ be a surjective morphism. Then there exists an \fo-partition \Kb
of $B^+$ such that for all $K \in \Kb$,

\begin{enumerate}
\item\label{it:c1} $\unclos{\beta(K)} \in \Sat(\Ss)$.
\item\label{it:c2} $K$ can be defined by a first-order formula
  of rank at most $|B|\cdot2^{|\unclos{\Ss}|^2}$.
\end{enumerate}
\end{proposition}

Before proving Proposition~\ref{prop:pumping}, we apply it to conclude
the proof of Proposition~\ref{prop:algoworks}. Set $\Ss =
\{\{\alpha(w)\} \mid w \in A^+\}$ and $\beta: A^+ \to \Ss$ defined by
$\beta(w) = \{\alpha(w)\}$ (note that $\beta$ is surjective). Recall
that by definition, $\Sat(\alpha) = \Sat(\Ss)$. From
Proposition~\ref{prop:pumping}, we obtain an \fo-partition \Kb of $A^+$
such that, for all $K \in \Kb$,
\begin{enumerate}
\item\label{it:c1alpha} $\alpha(K) = \unclos{\beta(K)} \in \Sat(\alpha)$.
\item\label{it:c2alpha} $K$ can be defined by a first-order formula
  of rank at most $|A|\cdot 2^{|S|^2}$.
\end{enumerate}

It is now immediate from Item~\ref{it:c1alpha} and the fact that
$\Sat(\alpha)$ is closed under downset that $\Is[\alpha](\Kb)
\subseteq \Sat(\alpha)$. We conclude that $\Is[\alpha] \subseteq
\Is[\alpha](\Kb) \subseteq \Sat(\alpha)$ which terminates the proof of
Proposition~\ref{prop:algoworks}. Moreover, we already know from Section~\ref{sec:correc} that
$\Is[\alpha] \supseteq \Sat(\alpha)$, so we actually obtain that
$\Is[\alpha] = \Is[\alpha](\Kb)$: \Kb is optimal for $\alpha$.

This is of particular interest. Indeed, we know from Theorem~\ref{thm:seppart}
that for any two languages recognized by $\alpha$ that are \fo-separable, one
can construct a separator as a union of languages in \Kb. Therefore, since our
proof of Proposition~\ref{prop:pumping} is constructive (\Kb is built by
induction), we obtain a method for constructing an \fo-separator for any pair
of \fo-separable languages recognized by $\alpha$. Finally, we know from Item~\ref{it:c2alpha} that
this separator has quantifier rank at most $|A|\cdot 2^{|S|^2}$ which yields the
second item in Theorem~\ref{th:main}.

\begin{corollary}[Second item in Theorem~\ref{th:main}]
Given two languages $L_0$ and $L_1$ that are recognized by $\alpha$,
if they are \fo-separable, then one can effectively construct an actual separator with
a formula of quantifier rank at most $|A|2^{|S|^2}$.
\end{corollary}

Note that a rough analysis of the procedure that constructs separators
yields a 2-\textsc{Exptime} upper bound on the complexity in the size of
$S$. This is because while the rank of the formula is ``only''
exponentially large in $S$, its size is one exponential larger in
general.

Another interesting remark is that while this paper is written from a
logical perspective (we prove that \Kb is an \fo-partition by
constructing a first-order formula for each language in \Kb), the
construction is not specific to first-order logic. In other words, the
construction could be easily adapted to obtain \emph{star-free
  expressions}, \emph{\ltl formulas} or \emph{counter-free automata}
defining the languages in \Kb.

\bigskip

It now remains to prove Proposition~\ref{prop:pumping}. The rest
of this section is devoted to this proof. We set $\Ss$ as a
subsemigroup of $2^S$ and $\beta : B^+ \rightarrow \Ss$ as a morphism as
in the statement of the proposition. The proof is a generalization of
Wilke's argument~\cite{wfo} for deciding first-order definability. As
explained above, the proof is constructive. We construct the partition
$\Kb$ as well as the first-order formulas that define its languages.
We proceed by induction on the following two parameters listed by
order of importance:
\begin{enumerate}[label=$(\alph*)$]
\item the index $|\unclos{\Ss}|$ of $\Ss$,
\item the size of $B$.
\end{enumerate}
The proof is divided into three cases:
\begin{itemize}
\item first, we consider the case when $|B|=1$.
\item otherwise, we distinguish two subcases, depending
on a property of $\beta$ called \emph{tameness}.
\end{itemize}

\subsection{\texorpdfstring{Special Case: $|B| = 1$.}{Special Case: |B| = 1.}}

In that case, $B$ is a singleton $\{b\}$. Hence all words are of the
form $b^n$ for some $n \geq 1$. It follows from a standard semigroup
theory argument that there exists $m \leq |\Ss| \leq
2^{|\unclos{\Ss}|}$ such that $\beta(b^m) = \beta(b^\omega)$. We
partition $B^+$ into $m$ languages.

For all $1 \leq i < m$, we set $K_i$ as the singleton $\{b^i\}$.
Finally, we set $K_m = \{b^j \mid j \geq m\}$. It is immediate by
definition that $\Kb = \{K_1,\dots,K_m\}$ is a partition of $B^+$.
Moreover, one can easily construct first-order formulas of rank at
most $m \leq 2^{|\unclos{\Ss}|}\leq 2^{|\unclos{\Ss}|^2}$ for $K_1,\dots,K_m$: \Kb is an
\fo-partition and we obtain Item~\ref{it:c2}. It remains to prove
Item~\ref{it:c1}.

For $i < m$, we have $\unclos{\beta(K_i)} = \beta(b^i) \in \Ss
\subseteq \Sat(\Ss)$. Hence Item~\ref{it:c1} is satisfied. Assume
now that $i = m$. By definition of $K_m$ and $\omega$,
\[
\unclos{\beta(K_m)} = \bigcup_{j \geq 0} \beta(b^{j+\omega}) =
(\beta(b)^\omega \cup \beta(b)^{\omega+1}) \cdots (\beta(b)^\omega \cup \beta(b)^{2\omega - 1})
\]
Since $\Sat(\Ss)$ is closed under multiplication, it suffices to
prove that for all $j$, we have $\beta(b)^\omega \cup \beta(b)^{\omega+j}
\in \Sat(\Ss)$. By definition, for any $j \geq 0$,
$\beta(b)^{\omega+j} \in \Ss \subseteq \Sat(\Ss)$. Moreover, observe that
\[
\beta(b)^\omega \cup \beta(b)^{\omega+j} = (\beta(b)^{\omega+j})^\omega \cup
(\beta(b)^{\omega+j})^{\omega+1}.
\]
Therefore, $\beta(b)^\omega \cup \beta(b)^{\omega+j} \in \Sat(\Ss)$ is
immediate by Operation~\ref{eq:oper}.

\medskip

This terminates the case $|B| = 1$. For the remainder of the proof, we
now assume that $|B| \geq 2$. As explained above, we distinguish two
cases depending on a property of $\beta$.

\medskip
\noindent
{\bf Tameness.} We say that $\beta$ is \emph{tame} if  $$\forall b \in
B,\ \unclos{\Ss} = \beta(b) \cdot \unclos{\Ss}\text{ and }\unclos{\Ss} =
\unclos{\Ss} \cdot \beta(b).$$

\subsection{\texorpdfstring{Case 1: $\beta$ is tame}{Case 1: beta is
    tame}} This is the base case: we don't use induction. We use
tameness to prove that $ \unclos{\Ss} \in \Sat(\Ss)$. Therefore, it
suffices to choose, $\Kb = \{B^+\}$ since $\unclos{\beta(B^+)}
= \unclos{\Ss}$ (recall that $\beta$ is assumed to be surjective).
This is a consequence of the following lemma:

\begin{lemma} \label{lem:basecase}
There exists a group $\Gs \subseteq \Ss$ such that $\unclos{\Gs} =
\unclos{\Ss}$.
\end{lemma}

We first use Lemma~\ref{lem:basecase} to finish the proof of this
case. Let $\Gs=\{T_1,\dots,T_n\}$ be a group as given by the
lemma. We prove that $\unclos{\Gs} \in \Sat(\Ss)$. Since $\Gs$ is a
group, we get $T_i^{\omega} = 1_{\Gs}$, so $T_i=T_1^{\omega}\cdots
T_{i-1}^{\omega}T_i^{\omega+1}T_{i+1}^{\omega} \cdots T_n^{\omega}$
for all~$i$. Combining these equalities gives us the equality
\[
\unclos{\Gs} = (T_1^{\omega} \cup T_1^{\omega+1}) \cdots (T_n^{\omega} \cup T_n^{\omega+1}).
\]
By definition, for all $i$, $T_i \in \Ss$, hence $T_i\in\Sat(\Ss)$ since $\beta$ is surjective. It then follows from
Multiplication Closure, \fo-Closure and the equality above that
$\unclos{\Gs} \in \Sat(\Ss)$. Since $\unclos{\Gs} =
\unclos{\Ss}$, this terminates the proof
in Case~1.

It remains to prove Lemma~\ref{lem:basecase}. We first prove that
while \Ss might not be a group itself, it is what we call a
\emph{pseudo-group}.

\medskip
\noindent
{\bf Pseudo-groups.} Let $\Ts$ be a subsemigroup of $2^{S}$. We say
that \Ts is a \emph{pseudo-group} if for all $T \in \Ts$,
$\unclos{\Ts} = T \cdot \unclos{\Ts}$ and $\unclos{\Ts} =
\unclos{\Ts} \cdot T $.

\begin{lemma} \label{lem:basecase2}
\Ss is a pseudo-group.
\end{lemma}

\begin{proof}
Set $T \in \Ss$. We prove that $\unclos{\Ss} = \unclos{\Ss} \cdot T$. The
equality  $\unclos{\Ss} = T \cdot
\unclos{\Ss}$  is
symmetrical. Since $\beta$ is surjective, there exists $w \in B^+$ such
that $T = \beta(w)$. We proceed by induction on the length of $w$. If
$w$ is of length $1$, this is by tameness of $\beta$.

Assume now that the result holds for words of length $k$ and that $w$
is of length $k+1$. This means that $w=ub$ with $u$ a word of length
$k$. By induction hypothesis, we get that $\unclos{\Ss} = \unclos{\Ss}
\cdot \beta(u)$. Moreover, using tameness, we get that $\unclos{\Ss} =
\unclos{\Ss} \cdot \beta(b)$. It follows that $\unclos{\Ss} =
\unclos{\Ss} \cdot \beta(u) \cdot \beta(b) = \unclos{\Ss} \cdot
\beta(w)$, which concludes the proof.
\end{proof}

We now finish the proof of Lemma~\ref{lem:basecase}. We prove that
any pseudo-group $\Ts \subseteq \Ss$ that is not already a group
strictly contains a subsemigroup $\Rs$ that remains a pseudo-group, and
such that $\unclos{\Rs} = \unclos{\Ts}$.  Applying this result iteratively
to \Ss yields the desired group \Gs.

\smallskip
Let $\Ts \subseteq \Ss$ be a pseudo-group that is not already a
group. An easy and standard argument implies that there
must exist $R \in \Ts$ such that $R \cdot \Ts \subsetneq \Ts$ or $\Ts
\cdot R \subsetneq \Ts$. By symmetry assume that it is the former and
set $\Rs = R \cdot \Ts$. By definition, $\Rs$ is closed under product
and is therefore a semigroup. It remains to prove that $\unclos{\Rs} =
\unclos{\Ts}$ and that $\Rs$ is a
pseudo-group.

By definition, we have $\unclos{\Rs} = R \cdot
\unclos{\Ts}$ and $R \cdot \unclos{\Ts} = \unclos{\Ts}$ since $\Ts$ is
a pseudo-group. We conclude that $\unclos{\Rs} =
\unclos{\Ts}$. Finally, set $RT \in \Rs$, we need to prove that
$\unclos{\Rs} = RT \cdot \unclos{\Rs}$ and $\unclos{\Rs} =
\unclos{\Rs} \cdot RT$. Both equalities are immediate since
$\unclos{\Rs} = \unclos{\Ts}$ and $\Ts$ is a pseudo-group.

\subsection{\texorpdfstring{Case 2: $\beta$ is not tame.}{Case 2: beta
    is not tame.}} This is the case where we use induction. By
hypothesis on $\beta$, there exists $b \in B$ such that $\unclos{\Ss}
\neq \beta(b) \cdot \unclos{\Ss}$ or $\unclos{\Ss} \neq \unclos{\Ss}
\cdot \beta(b)$. By symmetry, we assume the former, \emph{i.e.},
\[
\unclos{\Ss} \neq \beta(b) \cdot \unclos{\Ss}.
\]
We set $b$ as this
letter for the remainder of the proof.

\medskip
Recall that we have to construct an \fo-partition $\Kb$ of $B^+$
satisfying Items~\ref{it:c1} and~\ref{it:c2} in
Proposition~\ref{prop:pumping}. We begin by giving a brief overview of
the construction. Set $C = B \setminus \{b\}$. Observe that any word
$w \in B^+$ can be uniquely decomposed in three (possibly empty)
parts: a prefix in $C^+$, an infix in $(b^+C^+)^+$ and a suffix in
$b^+$. Our construction works by using induction to construct finite
partitions of the sets of possible prefixes, infixes and suffixes,
which yields a partition of the whole set $B^+$. For the sets of
prefixes and suffixes (\emph{i.e.}, $C^+$ and $b^+$), the partitions
are simply obtained by induction on the size of alphabet. For the set
of infixes (\emph{i.e.} $(b^+C^+)^+$) the argument is more involved
and is obtained by induction on the index of \Ss. We now make the
construction more precise. We begin by defining the partitions of the
sets of possible prefixes, infixes and suffixes, in the three lemmas
below.

\begin{lemma}[Partition of the prefixes] \label{lem:partpref}
There exists a finite partition $\Lb$ of $C^+$ such that for any
language $L \in \Lb$:
\begin{enumerate}
\item $\unclos{\beta(L)} \in \Sat(\Ss)$.
\item there exists a first-order formula of
  rank at most $(|B|-1) \cdot 2^{|\unclos{\Ss}|^2}$ that defines $L$.
\end{enumerate}
\end{lemma}

\begin{lemma}[Partition of the suffixes] \label{lem:partsuff}
There exists a finite partition $\Hb$ of $b^+$ such that for any
language $H \in \Hb$:
\begin{enumerate}
\item\label{it:c1induction2} $\unclos{\beta(H)} \in \Sat(\Ss)$.
\item\label{it:c2induction2} there exists a first-order formula of
  rank at most $2^{|\unclos{\Ss}|^2}$ that defines $H$.
\end{enumerate}
\end{lemma}

\begin{lemma}[Partition of the infixes] \label{lem:partinf}
There exists a finite partition $\Kb'$ of $(b^+C^+)^+$ such that
for any language $K \in \Kb'$:
\begin{enumerate}
\item $\unclos{\beta(K)} \in \Sat(\Ss)$.
\item there exists a first-order formula of
  rank at most $|B| \cdot 2^{|\unclos{\Ss}|^2} - 2$ that defines $K$.
\end{enumerate}
\end{lemma}

Lemmas~\ref{lem:partpref}, \ref{lem:partsuff} and~\ref{lem:partinf}
are proved using both our induction hypotheses. Before we present
these proofs, we use the three lemmas to construct the desired
\fo-partition $\Kb$ of $B^+$ and conclude the proof of
Proposition~\ref{prop:pumping}. We define $\Kb$ as follows:
\[
\Kb = \begin{array}{ll}
     & \{LK'H \mid L \in \Lb, K' \in \Kb' \text{ and } H \in \Hb\}\\
\cup & \{K'H \mid K' \in \Kb' \text{ and } H \in \Hb\}\\
\cup & \{LH \mid L \in \Lb \text{ and } H \in \Hb\}\\
\cup & \{LK' \mid L \in \Lb \text{ and } K' \in \Kb'\}\\
\cup & \Lb \cup \Kb' \cup \Hb
\end{array}
\]
That \Kb is partition of $B^+$ is immediate since $\Lb,\Kb'$ and $\Hb$
are partitions and any word in $B^+$ can be \emph{uniquely} decomposed
as the concatenation of a prefix in $b^+$, an infix in $(b^+C^+)^+$
and a suffix in $C^+$ (each one possibly empty, but not the three of them together). That \Kb is actually
an \fo-partition is a consequence of the following fact which describes
a standard construction for first-order logic over words.

\begin{fact} \label{fct:foconcat}
  Set $k \geq 0$ and two languages $L_1,L_2$, each defined by a
  first-order formula of rank at most $k$. Then $L_1L_2$ can be defined
  by a first-order formula of rank at most $k+1$.
\end{fact}

\begin{proof}
  A word is in $L_1L_2$ if and only it can be cut into a prefix in $L_1$
  and a suffix in $L_2$. Therefore, in first-order logic, it suffices to
  quantify existentially the position $x$ at which the cut is made and then
  use the formulas that define $L_1$ and $L_2$ (modified so that
  quantifications are restricted to the left or to the right of $x$) to
  test whether the prefix and suffix belong to $L_1$ and $L_2$. By
  construction, the formula we obtain has rank at most $k+1$.
\end{proof}

Since any language in \Kb is the concatenation of at most three
languages that are defined by first-order formulas of rank at most
$|B| \cdot 2^{|\unclos{\Ss}|^2} - 2$ (see the first items in the three
lemmas), we obtain that \Kb is an \fo-partition as well as
Item~\ref{it:c2} in Proposition~\ref{prop:pumping}.

Finally, Item~\ref{it:c1} in Proposition~\ref{prop:pumping} (\emph{i.e.}, for
all $K \in \Kb$, $\unclos{\beta(K)} \in \Sat(\Ss)$) is an immediate
consequence of Item~\ref{it:c1induction2} in the three lemmas and the
fact that $\Sat(\Ss)$ is closed under multiplication. This terminates
the proof of Proposition~\ref{prop:pumping}.

It now remains to prove Lemmas~\ref{lem:partpref}, \ref{lem:partsuff}
and~\ref{lem:partinf}. We first take care of Lemma~\ref{lem:partpref}
and Lemma~\ref{lem:partsuff} which are immediate by induction. Indeed,
Lemma~\ref{lem:partpref} is obtained by applying the induction hypothesis
on the second parameter (the size of the alphabet) to the restriction
of $\beta$ to $C^+$ (recall that $C = B \setminus\{b\}$). Furthermore,
Lemma~\ref{lem:partsuff} is exactly the special case when the alphabet
is of size one which is already proved.
\medskip

The proof of Lemma~\ref{lem:partinf} is more involved and is where we
use induction on the index of \Ss as well as our choice of the letter
$b$ (\emph{i.e.}, the fact that $\beta$ is not tame). We devote the remainder
of the section to this proof.

\newcommand\frB{\ensuremath{\mathfrak{B}}\xspace}
\newcommand\frL{\ensuremath{\mathfrak{L}}\xspace}
\newcommand\frS{\ensuremath{\mathfrak{S}}\xspace}
\newcommand\frb{\ensuremath{\mathfrak{b}}\xspace}
\newcommand\frw{\ensuremath{\mathfrak{u}}\xspace}
\newcommand\pred{\ensuremath{\text{pred}}\xspace}

Recall that our goal is to find a partition of $(b^+C^+)^+$ that
meets the conditions of the lemma.  We proceed in three steps.
First, we use Lemma~\ref{lem:partpref} and Lemma~\ref{lem:partsuff}
(\emph{i.e.}, our \fo-partitions of $b^+$ and $C^+$) to abstract the set
$b^+C^+$ as a finite alphabet $\frB$ and in turn the set $(b^+C^+)^+$
as the set of all words in $\frB^+$. This allows us to abstract the
restriction of $\beta$ to $(b^+C^+)^+$ as a semigroup morphism
$\gamma: \frB^+ \to \Ts$ into a new semigroup $\Ts \subseteq \Ss$.
Then in a second step, we use the fact that $\beta$ is not tame
(through our choice of $b$) to prove that $\Ts$ has smaller index than
\Ss. This enables us to apply induction to $\gamma$ and obtain an
\fo-partition of $\frB^+$. Finally, in the third step, we construct
the desired partition of $(b^+C^+)^+$ from that of~$\frB^+$.

\medskip
\noindent
{\bf Proof of Lemma~\ref{lem:partinf}, Step 1: Abstraction of
  $(b^+C^+)^+$.} We begin with the definition of the new alphabet
\frB. Intuitively, we want to simply set $\frB = \Hb \times \Lb$.
Indeed, we know by construction of \Hb and \Lb that $\{HL \mid H \in
\Hb \text{ and } L \in \Lb\}$ is a partition of $b^+C^+$. Therefore,
such an alphabet \frB would be a satisfying abstraction of $b^+C^+$.
However, there is an issue with this definition: in the proof, we do
not keep track of the size of the partitions \Hb and \Lb. Therefore,
such a definition does not allow us to control the size of \frB. This
is a problem for proving Item~\ref{it:c2induction2} in
Lemma~\ref{lem:partinf} as the bound on the quantifier rank of the
formulas obtained by induction depends on the size of the alphabet.
For this reason we use the following slightly different definition:
\[
\frB = \bigl\{\unclos{\beta(HL)} \mid H \in \Hb
\text{ and } L \in \Lb\bigr\} \subseteq 2^S.
\]
Observe that to any word $w \in b^+C^+$, one can associate a unique
letter $(w)_\frB \in \frB$: since $\Hb$ and $\Lb$ are partitions,
there exist unique $H \in \Hb$ and $L \in \Lb$ such that $w \in HL$,
we simply set $(w)_\frB = \unclos{\beta(HL)}$. This means that \frB
defines a finite partition of $b^+C^+$ (it is even an \fo-partition
by Fact~\ref{fct:foconcat}): two words are in the same class of the partition if they yield the same letter over $\frB$. We extend the definition to words $w \in
(b^+C^+)^+$: any such $w$ can be uniquely decomposed as $w = w_1
\cdots w_n$ with $w_1,\dots,w_n \in b^+C^+$, we set $(w)_\frB =
(w_1)_\frB \cdots (w_n)_\frB \in \frB^+$. In particular $\frB^+$
defines an infinite partition of $(b^+C^+)^+$.

We finish with the definition of the morphism $\gamma$. We set \Ts as
the subsemigroup of $2^S$ generated by $\frB$. Finally, set $\gamma:
\frB^+ \rightarrow \Ts$ defined by simply evaluating in \Ts the
product of the letters of a word in $\frB^+$. The following fact is
immediate from the definitions. It links $\gamma$ to $\beta$.

\begin{fact} \label{fct:linkmorph}
For any $w \in (b^+C^+)^+$, $\beta(w) \subseteq \gamma((w)_\frB)$.
\end{fact}

\medskip
\noindent
{\bf Proof of Lemma~\ref{lem:partinf}, Step 2: Constructing a
  partition of $\frB^+$.} We use induction to partition $\frB^+$. That
we may apply induction to $\gamma$  is a consequence of the following
fact, which is where we use our choice of $b$ (\emph{i.e.}, the fact that
$\beta$ is not tame).

\begin{fact} \label{fct:indexdown}
The index of $\Ts$ is strictly smaller than the index of \Ss.
\end{fact}

\begin{proof}
By definition, for any language $H \in \Hb$, we have
$\unclos{\beta(H)} \subseteq \unclos{\beta(b^+)}$. Hence,
\[
\unclos{\Ts} \subseteq \unclos{\beta(b^+)} \cdot \unclos{\Ss} =
\unclos{\beta(b^+) \cdot \Ss} \subseteq \beta(b) \cdot \unclos{\Ss}.
\]
By definition of $b$, we know that $\beta(b) \cdot \unclos{\Ss}
\subsetneq \unclos{\Ss}$. We conclude that $\unclos{\Ts} \subsetneq
\unclos{\Ss}$ which terminates the proof.
\end{proof}

It follows from Fact~\ref{fct:indexdown} that we may apply induction
on our first induction parameter (the index of \Ss) to $\gamma$ and
obtain an \fo-partition $\Fb$ of $\frB^+$ such that for all $F \in \Fb$:
\begin{enumerate}
\item $\unclos{\gamma(F)} \in \Sat(\Ts)$.
\item $F$ can be defined with a first-order formula of rank at most
  $|\frB|\cdot2^{|\unclos{\Ts}|^2}$.
\end{enumerate}

\medskip
\noindent
{\bf Proof of Lemma~\ref{lem:partinf}, Step 3: Constructing the
  partition $\Kb'$ of $(b^+C^+)^+$.} For any $F \in \Fb$, we define
$K_F = \{w \in (b^+C^+)^+ \mid (w)_\frB \in F\}$. Finally, we set
$\Kb' = \{K_F \mid F \in \Fb\}$. Since \Fb is a partition of
$\frB^+$, it is immediate that $\Kb'$ is a partition of
$(b^+C^+)^+$. It now remains to prove that Items~\ref{it:c1induction2}
and~\ref{it:c2induction2} in Lemma~\ref{lem:partinf} hold.

Let us first prove that Item~\ref{it:c1induction2} is satisfied. Set
$K \in \Kb'$. By definition, $K = K_F$ for some $F \in \Fb$. By
definition of $K_F$ and Fact~\ref{fct:linkmorph}, we have that,
\[
\unclos{\beta(K)} \subseteq \unclos{\gamma(F)}
\]
Moreover, by construction of $\Fb$, we know that $\unclos{\gamma(F)}
\in \Sat(\Ts)$. Finally, since $\Ts \subseteq \Ss$, we have $\Sat(\Ts)
\subseteq \Sat(\Ss)$. Using closure under downset, we obtain that
$\unclos{\beta(K)} \in \Sat(\Ss)$ which terminates the proof of
Item~\ref{it:c1induction2}.

It now remains to prove that Item~\ref{it:c2induction2} holds. Set $K
\in \Kb'$. We need to construct an \fo formula of rank at most $|B|
\cdot 2^{|\unclos{\Ss}|^2} - 2$ that defines $K$. Note that it is
immediate that $(b^+C^+)^+$ can be defined in \fo (with a formula of
rank $2$): this amounts to testing that the first letter in the word
is a $b$ and the last is a letter of $C$. Therefore it suffices to
construct a formula $\varphi_K$ such that for all $w \in (b^+C^+)^+$,
$w \models \varphi_K$ if and only if $w \in K$.

By construction, there exists $F \in \Fb$ such that $K = K_F$ as well as an
\fo formula $\Psi_F$ (over \frB) of rank less than $|\frB| \cdot
2^{|\unclos{\Ts}|^2}$ that defines $F$. By definition of $K = K_F$, it
suffices to construct $\varphi_K$ so that for any $w \in (b^+C^+)^+$,
\[
w \models \varphi_K \quad \text{if and only if} \quad (w)_\frB \models \Psi_F
\]
The construction is standard, we build $\varphi_K$ by modifying
$\Psi_F$. Consider a word $w \in (b^+C^+)^+$. We say that a position
$x$ in $w$ is \emph{distinguished} if and only if $x$ is labeled by
a ``$b$'' and position $(x+1)$ has label in $C$. In other words $x$ is
the rightmost $b$-labeled position of an infix in $b^+C^+$ of $w$.
Recall that by definition, every letter of $(w)_\frB$ abstracts an
infix in $b^+C^+$ of $w$. Therefore, one can associate a position
$\hat{x}$ of $(w)_\frB$ to every distinguished position $x$ of~$w$.

\begin{fact} \label{fct:defalpha}
For every $\frb \in \frB$, there exists a first-order formula
$\overline{\frb}(x)$ of rank at most $(|B|-1) \cdot
2^{|\unclos{\Ss}|^2}$ such that for any $w \in (b^+C^+)^+$ and any
distinguished position $x$ of $w$:
\begin{equation}
w,x \models \overline{\frb}(x) \text{ if and only if } w_\frB,
\hat{x} \models \frb(\hat{x}). \label{eq:eq1}
\end{equation}
\end{fact}

\begin{proof}
This amounts to testing whether the maximal infix in $b^+$ ending
at position $x$ in $w$ and the maximal infix in $C^+$ starting at
position $x+1$ in $w$ are in the appropriate languages of \Hb and
\Lb that yield letter $\frb$. This can easily be done with rank at most $(|B|-1) \cdot
2^{|\unclos{\Ss}|^2}$ since any language in \Hb or \Lb can be defined
by a formula of rank at most $(|B|-1) \cdot 2^{|\unclos{\Ss}|^2}$ (see
Lemma~\ref{lem:partpref} and Lemma~\ref{lem:partsuff}).
\end{proof}

The desired formula $\varphi_K$ is obtained from $\Psi_F$ by
restricting quantifications to distinguished positions and replacing
each atomic subformula of the form $\frb(x)$ by the formula
$\overline{\frb}(x)$. This can clearly be done in first-order
logic. Observe that this formula has rank at most $r =
|\frB|\cdot2^{|\unclos{\Ts}|^2} + (|B|-1) \cdot 2^{|\unclos{\Ss}|^2}$.
Since $|\frB| \leq 2^{|\unclos{\Ts}|}$ and $1 \leq |\unclos{\Ts}| \leq
|\unclos{\Ss}| - 1$, we obtain:
\[
r \leq |B|\cdot 2^{|\unclos{\Ss}|^2} - (2^{|\unclos{\Ss}|^2} -
2^{|\unclos{\Ss}|^2 - |\unclos{\Ss}|}) \leq |B|\cdot
2^{|\unclos{\Ss}|^2} -2
\]
Note that the last inequality is justified by the fact that
$|\unclos{\Ss}|\geq2$, which holds in this case since otherwise $\Ss$ would be the trivial group. This terminates the proof of Lemma~\ref{lem:partinf}.

\section{Alternate Algorithms}
\label{sec:altalgo}
In this section, we connect our algorithm with the ones of
Henckell~\cite{Henckell:Pointlike-sets:-finest-aperiodic:1988:a} and Henckell,
Rhodes and Steinberg~\cite{DBLP:journals/ijac/HenckellRS10a,qt}. These
algorithms (there are two of them) differ in the specific \fo-operation.
Although the change is minor and the correspondence between these algorithms
is easy to prove, this is what brings a complexity improvement, from
\textsc{Expspace} to \textsc{Exptime}. We include this easy section to bridge
the gap between all three algorithms. Note that the fact that the two given by
Henckell are equivalent was already shown
in~\cite{Henckell:Pointlike-sets:-finest-aperiodic:1988:a}.

The well-known decidable characterization of first-order logic by
Schützenberger~\cite{sfo,mnpfo} states that a language is
first-order \emph{definable} if and only if its syntactic semigroup is
\emph{aperiodic}. In the literature, there are many equivalent
definitions of aperiodicity. In this paper, we consider three of them:
one is equational, the second considers subgroups and the third
considers the \Hrel-classes. The relation `\Hrel' is one of Green's
relations which are well known in semigroup theory. Two elements $s,s'$
of a semigroup $S$ are \Hrel-equivalent if $s=s'$ or there exist $t^{}_\ell ,t'_\ell ,t^{}_r,
t'_r \in S$ such that $s t^{}_r = s'$, $s' t'_r = s$, $t^{}_\ell  s = s'$ and
$t'_\ell s'  = s$. We state the three equivalent definitions.

\begin{lemma}[Folklore, see~\cite{Pin13:MPRI}]
  A finite semigroup $S$ is aperiodic if and only if it satisfies one
  of the following equivalent statements:
  \begin{enumerate}
  \item for all $s \in S$, $s^\omega = s^{\omega+1}$.
  \item all subgroups in $S$ are trivial.
  \item all $\Hrel$-classes in $S$ are trivial.
  \end{enumerate}
\end{lemma}

Our saturation procedure $\Sat$ can be viewed as a generalization of
the first definition of aperiodicity. Indeed,
Operation~\ref{eq:oper} reflects the equation $s^\omega =
s^{\omega+1}$. In this section, we present two alternate and equivalent
saturation procedures that reflect the two other definitions. Let
$\alpha: A^+ \rightarrow S$ be a morphism into a finite semigroup.

\smallskip
Let $\Ss$ be a subsemigroup of $2^S$. We set $\Sat_{G}(\Ss)$ as the
smallest subset of $2^S$ that contains \Ss and is closed under
downset, multiplication and the following operation:
\begin{equation}
  \label{eq:oper-g}
  \text{for all $\Gs \subseteq \Sat_G(\Ss)$ that is a group, $\unclos{\Gs} \in \Sat_{G}(\Ss)$}.
\end{equation}
Similarly, $\Sat_{H}(\Ss)$ is the smallest subset of $2^S$ that
contains \Ss and is closed under downset, multiplication and the
following operation:
\begin{equation}
  \label{eq:oper-h}
  \text{for all $\Hs \subseteq \Sat_H(\Ss)$ that is an \Hrel-class, $\unclos{\Hs} \in \Sat_{H}(\Ss)$}.
\end{equation}

$\Sat_{G}$ reflects the second definition of aperiodicity and
$\Sat_{H}$ the third. In the following proposition, we state that the
three saturation procedures are equivalent and can therefore all be
used to compute $\Is[\alpha]$ by Proposition~\ref{prop:algoworks}.

\begin{proposition} \label{prop:alternate}
  Let $\Ss$ be a subsemigroup of $2^S$. Then, $$\Sat(\Ss) = \Sat_{G}(\Ss)
  = \Sat_{H}(\Ss).$$
\end{proposition}\smallskip

\noindent Note that the saturation procedure $\Sat_H$ is essentially Henckell's
original algorithm~\cite{Henckell:Pointlike-sets:-finest-aperiodic:1988:a}, where $\Sat_G$ was also shown to be a correct saturation operation. We
finish the section by proving Proposition~\ref{prop:alternate}.

\begin{proof}
We prove that $\Sat(\Ss) \subseteq \Sat_{H}(\Ss) \subseteq
\Sat_{G}(\Ss) \subseteq \Sat(\Ss)$. Let us first prove that $\Sat(\Ss)
\subseteq \Sat_{H}(\Ss)$.

\medskip\noindent
{\boldmath$\Sat(\Ss) \subseteq \Sat_{H}(\Ss)$.} By
definition of $\Sat(\Ss)$ and $\Sat_{H}(\Ss)$, this amounts to proving
that $\Sat_{H}(\Ss)$ is closed under \fo-closure: for any $T \in
\Sat_H(\Ss)$, $T^\omega \cup T^{\omega+1} \in \Sat_H(\Ss)$.

Set $T \in \Sat_H(\Ss)$. Observe that $T^{\omega+1}$ and $T^{\omega}$
are \Hrel-equivalent elements in the semigroup $\Sat_H(\Ss)$, and are
therefore both contained in some \Hrel-class $\Hs \subseteq
\Sat_H(\Ss)$. By definition of $\Sat_H(\Ss)$, we then have
$\unclos\Hs\in\Sat_H(\Ss)$. Hence, $T^{\omega} \cup
T^{\omega+1}\subseteq \unclos{\Hs} \in \Sat_H(\Ss)$, which ends the
proof since $\Sat_H(\Ss)$ is closed under downset.

\medskip\noindent
{\boldmath$\Sat_{H}(\Ss) \subseteq \Sat_{G}(\Ss)$}. This inclusion, which is
easy to prove, follows
from~\cite{Henckell:Pointlike-sets:-finest-aperiodic:1988:a}. We give here a
proof for the sake of completeness.

By definition of $\Sat_H(\Ss)$ and $\Sat_{G}(\Ss)$,
this amounts to proving that $\Sat_{G}(\Ss)$ is closed
under~\ref{eq:oper-h}: for all $\Hs \subseteq \Sat_G(\Ss)$ that is
an \Hrel-class, $\unclos{\Hs} \in \Sat_{G}(\Ss)$.

Let $\Hs \subseteq \Sat_G(\Ss)$ be an \Hrel-class. We claim that
either \Hs is a singleton, or there exists a group \Gs in
$\Sat_G(\Ss)$ and $R \in \Sat_G(\Ss)$ such that $\Hs = R \cdot \Gs$.
If \Hs is a singleton, then $\unclos{\Hs}$ is the unique element of
\Hs which belongs to $\Sat_G(\Ss)$. Otherwise, using closure under
multiplication, it follows that $\unclos\Hs = R\unclos\Gs \in
\Sat_G(\Ss)$ since $\unclos\Gs \in \Sat_G(\Ss)$ by
Operation~\ref{eq:oper-g} in the definition of $\Sat_G$.

It remains to prove the claim (which actually is not specific to
subsemigroups of a semigroup of subset): every \Hrel-class \Hs of
a semigroup \Ts is either a singleton, or of the form $R\cdot\Gs$, for
 $R\in\Ts$ and $\Gs$ a group in $\Ts$.  Let $\text{Stab} = \{T \in \Ts
 \mid \Hs\cdot T = \Hs \}$. If $\Hs$ is not a singleton, then Green's
 Lemma implies that $\text{Stab}$ is nonempty, and therefore it is a
 subsemigroup of \Ts. Let \Gs be an \Hrel-class of its minimal
 ideal. By standard results in semigroup theory~\cite[Chapter~V]{Pin13:MPRI}, \Gs is a group. Let us
 check that $\Hs = H\cdot\Gs$, for any $H\in\Hs$. Indeed, let $H\in\Hs$
 and let $E$ be the identity of \Gs. Since $E\in\text{Stab}$, we have
 $H=H'E$ for some $H' \in \Hs$, and so $HE = H$. Let now
 $H_1\in\Hs$. By definition, we have $H_1 = H\cdot X$ for some
 $X\in\Ts$. Note that since $\Gs$ is in the minimal ideal, we have
 $EXE\in \Gs$. Hence $H_1=H_1E=HXE=H(EXE)\in H\Gs$. This proves the
 claim and establishes the inclusion.

\medskip\noindent
{\boldmath$\Sat_{G}(\Ss) \subseteq \Sat(\Ss)$.} By definition of
$\Sat_G(\Ss)$ and $\Sat(\Ss)$, this amounts to proving that
$\Sat(\Ss)$ is closed under~\eqref{eq:oper-g}: for all $\Gs \subseteq
\Sat(\Ss)$ that is a group, $\unclos{\Gs} \in \Sat(\Ss)$.

Set $\Gs \subseteq \Sat(\Ss)$ that is a group and set
$\Gs=\{T_1,\dots,T_n\}$ with $T_i \in \Sat(\Ss)$ and let
$1_{\Gs}$ be the identity element of $\Gs$. Since $\Gs$ is a group,
for all $i$, $T_i^{\omega} = 1_{\Gs}$. In particular this means that
for all $i$, $T_i=T_1^{\omega}\cdots
T_{i-1}^{\omega}T_i^{\omega+1}T_{i+1}^{\omega} \cdots T_n^{\omega}$.
By combining these equalities, we get
\[
    \unclos{\Gs} =T_1\cup\cdots\cup T_n\subseteq (T_1^{\omega} \cup T_1^{\omega+1}) \cdots (T_n^{\omega} \cup T_n^{\omega+1}).
\]
It follows from \fo-closure and closure under multiplication that
$\unclos{\Gs} \in \Sat(\Ss)$.
\end{proof}

\section{Infinite Words}
\label{sec:omega}
An advantage of our technique for proving Theorem~\ref{th:main} is
that it generalizes smoothly to the setting of infinite words, \emph{i.e.},
it can be adapted to prove that \fo-separability is decidable for
infinite words. Both the algorithm itself and its proof are very similar to
those of the finite words setting. In particular, we retain all
results that we already have for finite words:

\begin{itemize}
\item we get an {\sc Exptime} upper bound on the complexity of the
  problem.
\item we get an exponential upper bound on the quantifier rank of a
  potential separator.
\item the proof is constructive: if a separator exists, one can be
  constructed by induction.
\end{itemize}

The remainder of the paper is devoted to the presentation of this
generalization. In this section, we introduce \iwords and generalize
our definitions to this setting: we define \ilangs, \isemis and
first-order logic over \iwords. We postpone the presentation of the
separation algorithm itself (which requires generalizing our framework
to \iwords) to the next section, Section~\ref{sec:omegasep}. Finally,
Section~\ref{sec:comp-i} is devoted to the proof of this algorithm.

\subsection{\texorpdfstring{Regular Languages of \iwords}{Regular Languages of
    infinite words}}$\quad$

\medskip
\noindent {\textbf{\iwords and \ilangs.}} Recall that $A$ is a finite
alphabet. We denote by $A^{\infty}$ the set of infinite words, called \emph{\iwords} over $A$. Note
that we still use the term ``word'' to mean an element of $A^+$. If
$u$ is a word and $v$ an \iword, we denote by $u \cdot v$ or $uv$ the \iword
obtained by concatenating $u$ to the left of $v$, and by $u^\infty$ the
\iword obtained by infinite concatenation of $u$ with
itself\,\footnote{In the literature, the \iword $u^\infty$ is usually
  denoted by $u^\omega$. Here, we use this non standard notation in
  order to avoid confusion with the idempotent power $\omega$ in
  semigroups.}. An \emph{\ilang} is a subset of $A^{\infty}$. \emph{Regular}
\ilangs are those that are accepted by \emph{nondeterministic Büchi
  automata} (NBA). Again, we will only work with the algebraic
representation of \ilangs that we recall below.

\medskip
\noindent
{\textbf{\isemis.}} We briefly recall the definition of \isemis, which
play the role of semigroups in the setting of \iwords. For more details, we
refer the reader to~\cite{Perrin&Pin:Infinite-Words:2004:a}.

An \emph{\isemi} is a pair $(S_+,S_{\infty})$ where $S_+$ is a
semigroup and $S_{\infty}$ is a set. Moreover, $(S_+,S_{\infty})$ is
equipped with two additional products: a \emph{mixed product} $S_+
\times S_{\infty} \rightarrow S_{\infty}$ that maps $s,t \in
S_+,S_{\infty}$ to an element denoted $st$, and an \emph{infinite
  product} $(S_{+})^{\infty} \rightarrow S_{\infty}$ that maps an
infinite sequence $s_1,s_2,\dots \in (S_{+})^{\infty}$ to an element
of $S_{\infty}$ denoted by $s_1s_2\cdots$. We require these products
as well as the semigroup product of $S_+$ to satisfy all possible
forms of associativity (see~\cite{Perrin&Pin:Infinite-Words:2004:a}
for details). Finally, we denote by $s^{\infty}$ the element
$sss\cdots$. Observe that $(A^+,A^{\infty})$ is an \isemi.

The notions of subsemigroups and morphisms can be adapted to
\isemis. In particular, if $T_+$ is a subsemigroup of $S_+$ and
$T_\infty$ is the set obtained by applying the infinite product to all
sequences of $T_+$, then $(T_+,T_{\infty})$ is a \sisemi of
$(S_+,S_\infty)$ called the \emph{\sisemi generated by~$T_+$}.

An \isemi is said to be \emph{finite} if both $S_+$ and $S_{\infty}$
are finite. Note that even if an \isemi is finite, it is not obvious
that a finite representation of the infinite product exists. However,
it was proven by Wilke~\cite{womega} that the infinite product is
fully determined by the mapping $s \mapsto s^\infty$, yielding
a finite representation for finite \isemis. An \ilang~$L$ is said to be
\emph{recognized} by an \isemi $(S_+,S_{\infty})$ if there exists $F
\subseteq S_{\infty}$ as well as a morphism $\alpha : (A^+,A^\infty)
\rightarrow (S_+,S_\infty)$ such that $L = \alpha^{-1}(F)$. It is well known
that an \ilang is regular if and only if it is recognized by a
\emph{finite} \isemi. Moreover~\cite{womega}, from any NBA recognizing
$L$, one can compute a canonical smallest \isemi recognizing $L$,
called the \emph{syntactic \isemi of $L$.}

As for finite words, when working on separation, it is convenient to
consider a single recognizing object for both input languages rather
than two separate objects. Again, this is not restrictive: given two
\ilangs and two associated recognizing \isemis, one can define (and
compute) a single \isemi that recognizes both languages by taking the
Cartesian product of the two original \isemis.

\medskip
\noindent
{\bf Semigroup of Subsets.} For an \isemi $(S_+,S_{\infty})$, note
that $(2^{S_+},2^{S_\infty})$ is an \isemi with the products
defined in the natural way.

\subsection{\texorpdfstring{First-Order Logic over
    \iwords.}{First-Order Logic over Infinite Words.}}$\quad$

As for words, an \iword can be viewed as a sequence of positions that
are labeled over $A$ (the difference being that in the case of
\iwords, the sequences are infinite: there is a leftmost position but
no rightmost one). Therefore, first-order formulas as we defined them can
also be interpreted on \iwords and we can simply say that an \ilang
$L \subseteq A^\infty$ is first-order definable if and only if there
exists an \fo formula $\varphi$ such that $L = \{w \in A^\infty \mid w
\models \varphi\}$.

First-order logic over \iwords shares similar properties with
first-order logic over words. First, the equivalence with star-free
languages still holds for \ilangs: an \ilang is first-order definable
if and only if it is
star-free~\cite{Ladner:Application-model-theoretic-games:1977:a,tfo}.
Furthermore, Schützenberger's Theorem was generalized to \ilangs by
Perrin~\cite{pfo}: a regular \ilang is star-free (and hence \fo) if
and only if the finite semigroup $S_+$ of its syntactic \isemi
$(S_+,S_\infty)$ is aperiodic. Note that we obtain an alternate proof
of this theorem as a simple consequence of our separation algorithm.

Our main theorem for \ilangs is similar to Theorem~\ref{th:main} for
languages and is as~follows.

\begin{theorem} \label{th:main-i}
  Let $L_0,L_1$ be regular \ilangs recognized by a morphism
  $\alpha: (A^+,A^\infty) \to (S_+,S_\infty)$ into a finite \isemi. The
  two following items hold.
  \begin{enumerate}
  \item One can decide in {\sc Exptime} with respect to $|S_+|$ whether
    $L_0$ and $L_1$ are \fo-separable.
  \item When $L_0$ and $L_1$ are \fo-separable, one can construct an
    actual separator with a formula of quantifier rank at most
    $|A|2^{|S_+|^2}+1$.
  \end{enumerate}
\end{theorem}\smallskip

\noindent The proof of Theorem~\ref{th:main-i} is very similar to the one of
Theorem~\ref{th:main} and relies on the same objects: \fo-partitions
and optimal imprints (generalized to \iwords). In particular, this
means that our proof remains constructive: it yields an inductive way
to construct an actual separator, \emph{i.e.}, an \fo-partition of $A^\infty$
that is optimal for the input morphism $\alpha$, when it exists (a
rough analysis yields a $2$-{\sc Exptime} complexity in $|S_+|$).

\section{Separation Algorithm for Infinite Words}
\label{sec:omegasep}
In this section, we present our separation algorithm for first-order
logic over \iwords. As we explained, this algorithm is based on a
generalization of our finite words framework to \iwords. Therefore,
we divide this section in two parts. In the first part, we generalize
\fo-partitions, imprints and Theorem~\ref{thm:seppart} to
\iwords. Then, in the second part, we present our separation algorithm.

\subsection{Definition}

Recall that a finite alphabet $A$ is fixed. For the definitions, we
let \Cs as an arbitrary class consisting of languages and \ilangs (\emph{i.e.}, $\Cs
\subseteq 2^{A^+} \cup 2^{A^\infty}$). Moreover, we assume that
\begin{itemize}
\item when restricted to languages, \Cs is nonempty, closed
  under Boolean operations and quotients, and contains only regular languages,
\item when restricted to \ilangs, \Cs is nonempty and closed
  under Boolean operations.
\end{itemize}

One can verify that \fo satisfies these conditions. Note that since
\Cs is assumed to contain both languages and \ilangs, one can consider
two kinds of \Cs-partitions: \Cs-partitions of $A^+$ and
\Cs-partitions of $A^\infty$.

Set $\alpha: (A^+,A^\infty) \to (S_+,S_\infty)$ as an arbitrary
morphism into a finite \isemi $(S_+,S_\infty)$. Observe that any such
morphism $\alpha$ can be decomposed into two maps: a morphism
$\alpha_+: A^+ \to S_+$ into a finite semigroup $S_+$ and a map
$\alpha_\infty: A^\infty \to S_\infty$ into a finite set $S_\infty$.

Since $\alpha_+$ is a morphism, we may directly apply our
definition of imprints for finite words to it: if \Kb is a \Cs-partition of $A^+$, then
$\Is[\alpha_+](\Kb) \subseteq 2^{S_+}$ is well-defined. Similarly, by
hypothesis on \Cs, the optimal \Cs-partitions of $A^+$ for $\alpha_+$
are well-defined as those having the smallest possible imprint on
$\alpha_+$: $\Is_\Cs[\alpha_+]$. In particular, we know from
Lemma~\ref{lem:usemi} and our hypothesis on \Cs that $\Is_\Cs[\alpha_+]$
is a subsemigroup of $2^{S_+}$.

It turns out that aside from Lemma~\ref{lem:usemi}, these definitions
do not require $\alpha_+$ to be a semigroup morphism. Hence, they can
also be applied to $\alpha_\infty$. If \Kb is a \Cs-partition of $A^\infty$,
then the imprint of $\Kb$ on $\alpha_\infty$ is defined by,
\[
\Is[\alpha_\infty](\Kb)  = \{T \in 2^{S_\infty} \mid \text{there exists $K \in \Kb$
  such that $T \subseteq \alpha_\infty(K)$}\} \subseteq 2^{S_\infty}.
\]
Note that one can verify that imprints on $\alpha_\infty$ still verify
Fact~\ref{fct:trivial}  (\emph{i.e.}, for all $w \in A^\infty$,
$\{\alpha_\infty(w)\} \in \Is[\alpha_\infty](\Kb)$) and
Fact~\ref{fct:downclos} (\emph{i.e.}, $\Is[\alpha_\infty](\Kb)$ is closed
under downset). Finally, the optimal \Cs-partitions of $A^\infty$ for
$\alpha_\infty$ are defined as those having the smallest possible
imprint on $\alpha_\infty$: $\Is_\Cs[\alpha_\infty]$ (as before, we
need the fact that \Cs is closed under intersection to prove that
there exists at least one optimal \Cs-partition, see
Lemma~\ref{lem:optimal}).

\begin{remark}\label{rem:subisemi}
Note that in this case, since $S_\infty$ is not a semigroup, it is not
true that $\Is_\Cs[\alpha_\infty]$ is a semigroup. However, with
additional hypotheses on \Cs (which correspond to the usual
generalization of closure under quotients to classes of \ilangs), one
could prove that the pair $(\Is_\Cs[\alpha_+],\Is_\Cs[\alpha_\infty])$
is a \sisemi of $(2^{S_+},2^{S_\infty})$. We will prove this property
in the special case where $\Cs = \fo$.
\end{remark}

We can now generalize Theorem~\ref{thm:seppart} to \iwords.

\begin{theorem} \label{thm:sepparti}
Let $\alpha: (A^+,A^\infty) \to (S_+,S_ \infty)$ be a morphism into a
finite \isemi $(S_+,S_\infty)$. Let $L_1,L_2 \subseteq A^\infty$ be two
\ilangs recognized by $\alpha$ and let $T_1,T_2 \subseteq S_\infty$
be the corresponding accepting sets. The following properties are
equivalent:
\begin{enumerate}
\item\label{item:cov2sep1-inf} $L_1$ and $L_2$ are \Cs-separable.
\item\label{item:cov2sep2-inf} for all $t_1 \in T_1$ and all $t_2 \in T_2$, $\{t_1,t_2\}
  \not\in \Is_\Cs[\alpha_\infty]$.
\item\label{item:cov2sep3-inf} for any \Cs-partition \Kb of $A^\infty$
  that is optimal for $\alpha_\infty$, $L_1$ and $L_2$ are separable by a
  union of languages in $\Kb$.
\end{enumerate}\smallskip
\end{theorem}

\noindent The proof of Theorem~\ref{thm:sepparti} is identical to that of
Theorem~\ref{thm:seppart}. In view of the theorem, generalizing our
approach to \ilangs amounts to finding an algorithm that computes
$\Is_\Cs[\alpha_\infty]$ from a morphism $\alpha$ into a finite
\isemi. We now present such an algorithm.

\subsection{Separation Algorithm}

We can now generalize our separation algorithm to the setting of
\iwords. Let $\alpha: (A^+,A^\infty) \to (S_+,S_\infty)$ be a morphism
into a finite \isemi $(S_+,S_\infty)$. From now on, we only work with
the class \fo, therefore, we simply write
$(\Is[\alpha_+],\Is[\alpha_\infty])$ to denote the pair
$(\Is_\fo[\alpha_+],\Is_\fo[\alpha_\infty])$. We present an algorithm for
computing this pair.

We already know how to compute $\Is[\alpha_+]$ from
$\alpha_+$: $\Is[\alpha_+] = \Sat(\Ss_+)$ with ``$\Sat$''
defined in Section~\ref{sec:main} (see
Proposition~\ref{prop:algoworks}). It turns out that
$\Is[\alpha_\infty]$ can easily be computed from~$\Is[\alpha_+]$.
For $\Ss \subseteq 2^{S_+}$, let $\Sat_\infty(\Ss)$ be the smallest subset of $2^{S_\infty}$ closed
under the following operations:
\begin{enumerate}
\item\label{eq:ioper1} For any $T \in \Ss$, we have $T^\infty
\in \Sat_\infty(\Ss)$.
\item\label{eq:ioper2} For any $T \in \Ss$ and $T' \in
  \Sat_\infty(\Ss)$, we have $TT' \in \Sat_\infty(\Ss)$.
\item\label{eq:ioper3} $\Sat^\infty(\Ss)$ is closed under downset:
$\Sat_\infty(\Ss) = \downclos \Sat_\infty(\Ss)$.
\end{enumerate}

In other words, $\Sat_\infty(\Ss)$ is the smallest subset of $2^{S_\infty}$
that is closed under downset and such that $(\Ss,\Sat_\infty(\Ss))$ is a
\sisemi of $(2^{S_+},2^{S_\infty})$. This smallest subset of $2^{S_\infty}$
clearly exists. Finally, we set $\Sat_\infty(\alpha)$ as
$\Sat_\infty(\Is[\alpha_+])$.

\begin{proposition} \label{prop:algoworks-i}
Set $\alpha: (A^+,A^\infty) \to (S_+,S_\infty)$ as a morphism
into a finite \isemi $(S_+,S_\infty)$. Then,
\[
\Is[\alpha_\infty] = \Sat_\infty(\alpha).
\]
\end{proposition}\smallskip

\noindent Since we already know how to compute $\Is[\alpha_+]$ in {\sc Exptime}
with respect to $|S_+|$ (see Proposition~\ref{prop:algoworks}), it
follows from Proposition~\ref{prop:algoworks-i} that one can compute
$\Is[\alpha_\infty]$ in {\sc Exptime} with respect to $|S_+|$ as well.
It then follows from Theorem~\ref{thm:sepparti} that this generalizes
our upper bound on the complexity of the separation problem to
\ilangs: one can decide in {\sc Exptime} whether two \ilangs are
\fo-separable. Therefore, we obtain the first item in
Theorem~\ref{th:main-i} as a corollary. We will obtain the second item
as a byproduct of the proof of Proposition~\ref{prop:algoworks-i}.

Another important remark is that it follows from
Proposition~\ref{prop:algoworks-i} that
$(\Is[\alpha_+],\Is[\alpha_\infty])$ is a \sisemi of
$(2^{S_+},2^{S_\infty})$. As explained in Remark~\ref{rem:subisemi},
this property is not specific to \fo. On the other hand, what is
specific to \fo is that $\Is[\alpha_\infty]$ is the {\bf smallest}
subset of $2^{S^\infty}$ that is closed under downset and such that
$(\Is[\alpha_+],\Is[\alpha_\infty])$ is a \sisemi of
$(2^{S_+},2^{S_\infty})$.

Finally, a consequence of
Proposition~\ref{prop:algoworks-i} is that we obtain Perrin's
theorem~\cite{pfo} as a corollary, just as we obtained Schützenberger's
one~\cite{sfo} as a corollary of Proposition~\ref{prop:algoworks}.

\begin{corollary} \label{cor:caraci}
Let $L$ be a regular \ilang. Then $L$ can be defined in \fo if and
only if its syntactic \isemi $(S_+,S_\infty)$ is such that $S_+$ is
aperiodic.
\end{corollary}

\begin{proof}
The proof is similar to that of Corollary~\ref{cor:carac}. It is
known that an \ilang is definable in \fo if and only if all languages
and \ilangs recognized by its syntactic \isemi are definable in \fo as
well (as before, this is actually not specific to \fo and true for all
classes of \ilangs that are ``Varieties'',
see~\cite{Perrin&Pin:Infinite-Words:2004:a} for example). It follows
that, if $\alpha: (A^+,A^\infty) \to (S_+,S_\infty)$ is the syntactic
\isemi of $L$, then $L$ is definable in \fo if and only if $\Is[\alpha_+]$
and $\Is[\alpha_\infty]$ contain only singletons and the empty set.
One can then verify from Proposition~\ref{prop:algoworks} that this is
is equivalent to $S_+$ satisfying, $s^\omega = s^{\omega+1}$ for all $s
\in S_+$.
\end{proof}

\noindent
It now remains to prove Proposition~\ref{prop:algoworks-i}. We present
this proof in the next section, Section~\ref{sec:comp-i}.

\section{Correctness of the Infinite Words Algorithm}
\label{sec:comp-i}
This section is devoted to the proof
Proposition~\ref{prop:algoworks-i}. We fix a morphism $\alpha:
(A^+,A^\infty) \to (S_+,S_\infty)$ into a finite \isemi
$(S_+,S_\infty)$ for the whole section. We have to prove that
$\Is[\alpha_\infty] = \Sat_\infty(\alpha)$. We separate the proof
in two parts, each one corresponding to an~inclusion.

\subsection{Soundness of the Algorithm}
\label{sec:corr-algor}
We begin with the easiest inclusion: $\Sat_\infty(\alpha) \subseteq
\Is[\alpha_\infty]$. This corresponds to soundness of the algorithm:
it only computes sets belonging to $\Is[\alpha_\infty]$. By
definition of $\Sat_\infty(\alpha)$, we need to prove that:
\begin{itemize}
\item For any $T \in \Is[\alpha_+]$, we have $T^\infty \in \Is[\alpha_\infty]$.
\item For any $T \in \Is[\alpha_+]$ and $T' \in \Is[\alpha_\infty]$, we have
  $TT' \in \Is[\alpha_\infty]$.
\item $\Is[\alpha_\infty]$ is closed under downset.
\end{itemize}

That $\Is[\alpha_\infty]$ is closed under downset is immediate from
the definition ($\Is[\alpha_\infty]$ is an imprint). We prove the two
other items. The proof relies on the generalization of the equivalence
$\foeq{k}$ to \iwords: given two \iwords $w,w' \in A^\infty$ and $k
\in \nat$, we write $w \foeq{k} w'$ to denote the fact that $w$ and
$w'$ satisfy the same formulas of quantifier rank $k$. One can verify
that Lemma~\ref{lem:optequ} still holds for \iwords.

\begin{lemma} \label{lem:optequi}
  Let $T \in 2^{S_\infty}$. Then $T \in
  \Is[\alpha_\infty]$ if and only if for all $k \in \nat$, there exists
  an equivalence class $W \subseteq A^\infty$ of \foeq{k} such that $T
  \subseteq \alpha(W)$.
\end{lemma}

We can now finish the proof of soundness. Set $T \in \Is[\alpha_+]$
and $T' \in \Is[\alpha_\infty]$. We use Lemmas~\ref{lem:optequi} to
prove that $T^\infty \in \Is[\alpha_\infty]$ and $TT' \in
\Is[\alpha_\infty]$. Set $k \in \nat$.

By Lemmas~\ref{lem:optequ} and~\ref{lem:optequi}, we obtain an
equivalence class $W \subseteq A^+$ of \foeq{k} (over finite words)
and an equivalence class $W' \subseteq A^\infty$ of \foeq{k} (over
\iwords) such that $T \subseteq \alpha(W)$ and $T' \subseteq
\alpha(W')$. We know from Lemma~\ref{lem:optequi} that it suffices to
prove that $W^\infty$ and $WW'$ are included in equivalence classes of
\foeq{k} in order to conclude that $T^\infty \in \Is[\alpha_\infty]$ and
$TT' \in \Is[\alpha_\infty]$. This can be easily verified using a
generalization of the first item of Lemma~\ref{lem:efconcat} to
\iwords: for any $w \in W$ and $w' \in W'$, one can verify that
any \iword in $W^\infty$ is $\kfoeq$-equivalent to $w^\infty$ and that
any \iword in $WW'$ is $\kfoeq$-equivalent to $ww'$.

\subsection{Completeness of the Algorithm}
\label{sec:compl-algor}

We now turn to the most interesting inclusion in
Proposition~\ref{prop:algoworks-i}: $\Is[\alpha_\infty] \subseteq
\Sat_\infty(\alpha)$. The proof is a generalization of that of
Proposition~\ref{prop:pumping} to the setting of \iwords. In
particular, the proof remains constructive: we use induction to
construct an \fo-partition \Kb of $A^\infty$ whose imprint on
$\alpha_\infty$ is included in $\Sat_\infty(\alpha)$. This proves that
$\Is[\alpha_\infty] \subseteq \Is[\alpha_\infty](\Kb) \subseteq
\Sat_\infty(\alpha)$. The induction is stated in the following
proposition.

\begin{proposition} \label{prop:pumping-i}
Let $(\Ss_+,\Ss_\infty)$ be a \sisemi of\/ $(2^{S_+},2^{S_\infty})$ and let $\beta:
(B^+,B^\infty) \to  (\Ss_+,\Ss_\infty)$ be a surjective morphism. Then
there exists an \fo-partition \Kb of $B^\infty$ such that for all $K
\in \Kb$:

\begin{enumerate}
\item\label{it:c1-i} $\unclos{\beta(K)} \in \Sat_\infty(\Sat(\Ss_+))$.
\item\label{it:c2-i} $K$ can be defined by a first-order formula
  of rank at most $|B|\cdot2^{|\unclos{\Ss}|^2}+1$.
\end{enumerate}
\end{proposition}

Let us first use Proposition~\ref{prop:pumping-i} to conclude the
proof of Proposition~\ref{prop:algoworks-i}. Set $\Ss_+ =
\{\{\alpha(w)\} \mid w \in A^+\}$, $\Ss_\infty = \{\{\alpha(w)\} \mid
w \in A^\infty\}$ and $\beta: (B^+,B^\infty) \to  (\Ss_+,\Ss_\infty)$
defined by $\beta(w) = \{\alpha(w)\}$ for $w \in A^+ \cup A^\infty$
(note that $\beta$ is surjective). Recall that we already know from
Proposition~\ref{prop:algoworks} that $\Is[\alpha_+] =
\Sat(\Ss_+)$. Therefore, by definition, $\Sat_\infty(\alpha) =
\Sat_\infty(\Sat(\Ss_+))$. From Proposition~\ref{prop:pumping-i}, we
obtain an \fo-partition \Kb of $A^\infty$ such that  for all $K \in
\Kb$,
\begin{enumerate}
\item $\alpha(K) = \unclos{\beta(K)} \in \Sat_\infty(\alpha)$.
\item any $K \in \Kb$ can be defined by a first-order formula
  of rank at most $|A|\cdot 2^{|S_+|^2}+1$.
\end{enumerate}

It is now immediate from Item~1 and the fact that
$\Sat_\infty(\alpha)$ is closed under downset that
$\Is[\alpha_\infty](\Kb) \subseteq \Sat_\infty(\alpha)$. We conclude
that $\Is[\alpha_\infty] \subseteq \Is[\alpha_\infty](\Kb) \subseteq
\Sat_\infty(\alpha)$ which terminates the proof of
Proposition~\ref{prop:algoworks-i}. Moreover, since we already know
that $\Sat_\infty(\alpha) \subseteq \Is[\alpha_\infty]$, we actually
have $\Is[\alpha_\infty] = \Is[\alpha_\infty](\Kb)$: $\Kb$ is optimal
for $\alpha_\infty$. Therefore, we obtain the second item in
Theorem~\ref{th:main-i} from Item~\ref{it:c2-i} of Proposition~\ref{prop:pumping-i}.

\begin{corollary}[Second item in Theorem~\ref{th:main-i}]
Given two \ilangs $L_0$ and $L_1$ that are recognized by $\alpha$,
if they are \fo-separable, then one can construct an actual separator with
a formula of quantifier rank at most $|A|2^{|S_+|^2}+1$.
\end{corollary}
\enlargethispage{\baselineskip}

It remains to prove Proposition~\ref{prop:pumping-i}. We generalize
the techniques we used to prove Proposition~\ref{prop:pumping}. Note
that in several cases, the construction will require building an
\fo-partition of $B^+$ (or of a subset of $B^+$). In this cases, we
will simply use Proposition~\ref{prop:pumping}. As for
Proposition~\ref{prop:pumping}, we construct $\Kb$ by induction on
the following two parameters listed by order of importance:
\begin{enumerate}[label=$(\alph*)$]
\item the index $|\unclos{\Ss_+}|$ of $\Ss_+$,
\item the size of $B$.
\end{enumerate}

Observe that the case $|B| = 1$ is trivial in this setting: in that
case $B^\infty$ is a singleton. We now assume that $|B| > 1$
and distinguish two subcases, depending on whether the restriction of
$\beta$ to $B^+$ is \emph{tame}. Recall that we say that $\beta$ is
\emph{tame} if for all $b \in B$, $\unclos{\Ss_+} = \beta(b) \cdot
\unclos{\Ss_+}$ and  $\unclos{\Ss_+} = \unclos{\Ss_+} \cdot \beta(b)$.

\medskip
\noindent
{\bf Case 1: $\beta$ is tame.} As we have seen in the proof of
Proposition~\ref{prop:pumping}, in that case we have $\unclos{\Ss_+}
\in \Sat(\Ss_+)$. By surjectivity of $\beta$ it is immediate that
$(\unclos{\Ss_+})^\infty = \unclos{\Ss_\infty}$. Therefore, by Item~\ref{eq:ioper1}
in the definition of $\Sat_\infty$, $\unclos{\Ss_\infty} \in
\Sat_\infty(\Sat(\Ss_+))$. It is therefore sufficient to set $\Kb =
\{B^\infty\}$ to satisfy Item~\ref{it:c1-i} and Item~\ref{it:c2-i}
in the proposition.

\medskip
\noindent
{\bf Case 2:  $\beta$ is not tame.}  By
hypothesis on $\beta$, there exists $b \in B$ such that $\unclos{\Ss_+}
\neq \beta(b) \cdot \unclos{\Ss_+}$ or $\unclos{\Ss_+} \neq
\unclos{\Ss_+} \cdot \beta(b)$. By symmetry, we assume the former,
\emph{i.e.}, $\unclos{\Ss_+} \neq \beta(b) \cdot \unclos{\Ss_+}$. We
set $b$ as this letter for the remainder of the proof.

Recall that we have to construct an \fo-partition $\Kb$ of $B^\infty$
satisfying Items~\ref{it:c1-i} and~\ref{it:c2-i} in
Proposition~\ref{prop:pumping-i}. Set $C = B \setminus \{b\}$ and
observe that $B^\infty$ is the (disjoint) union of the following five
sets:
\begin{equation} \label{eq:fivesets}
B^\infty = b^\infty \cup B^*Cb^\infty \cup C^\infty
\cup B^*bC^\infty
\cup C^*(b^+C^+)^\infty.
\end{equation}
Therefore, it suffices to find \fo-partitions satisfying
Items~\ref{it:c1-i} and~\ref{it:c2-i} for all five sets to obtain
the desired partition of $B^\infty$. These partitions are defined from
\fo-partitions of $C^+$ and $B^+$ (obtained from
Proposition~\ref{prop:pumping}), of $b^\infty$ and $C^\infty$
(obtained by induction on $|B|$ in Proposition~\ref{prop:pumping-i})
and of $C^*(b^+C^+)^\infty$ (obtained by induction on the index of
$\Ss_+$ in Proposition~\ref{prop:pumping-i}). Since the construction
is similar for all five sets, we only detail the case of
$C^*(b^+C^+)^\infty$ (other cases are handled similarly). The
construction is based on the following two lemmas.

\begin{lemma}[Partition of $C^+$] \label{lem:partcp}
There exists a finite partition $\Lb$ of $B^+$ such that for any
language $L \in \Lb$:
\begin{enumerate}
\item $\unclos{\beta(L)} \in \Sat(\Ss_+)$.
\item there exists a first-order formula of
  rank at most $|C| \cdot 2^{|\unclos{\Ss_+}|^2}$ that defines $L$.
\end{enumerate}
\end{lemma}

\begin{lemma}[Partition of $(b^+C^+)^\infty$] \label{lem:partbc}
There exists a finite partition $\Kb'$ of $(b^+C^+)^\infty$ such
that for any language $K \in \Kb'$:
\begin{enumerate}
\item $\unclos{\beta(K)} \in \Sat_\infty(\Sat(\Ss_+))$.
\item there exists a first-order formula of
  rank at most $|B| \cdot 2^{|\unclos{\Ss_+}|^2}$ that defines $K$.
\end{enumerate}
\end{lemma}

Lemma~\ref{lem:partcp} is obtained by applying
Proposition~\ref{prop:pumping} to the restriction of $\beta$ to
$C^+$. The proof of Lemma~\ref{lem:partbc} is a straightforward
generalization to \iwords of the proof of Lemma~\ref{lem:partinf} and
is left to the reader (note that this is where our choice of $b$ and
induction on the index of $\Ss_+$ are used).

Let us now explain how to construct the desired \fo-partition of
$C^*(b^+C^+)^\infty$. Consider the following partition $\Kb''$ of
$C^*(b^+C^+)^\infty$,
\[
\Kb'' = \{K' \mid K' \in \Kb'\} \cup \{LK' \mid L \in \Lb \text{ and
  } K' \in \Kb' \}.
\]
It is immediate from the fact that $\Lb$ and $\Kb'$ are
partitions that $\Kb''$ is a partition of
$C^*(b^+C^+)^\infty$. Moreover, it follows from Item~(1) of
Lemmas~\ref{lem:partcp} and~\ref{lem:partbc} and the second item in
the definition of $\Sat_\infty$ that $\Kb''$ satisfies the first item
in Proposition~\ref{prop:pumping-i}: for all $K'' \in \Kb''$,
$\unclos{\beta(K'')} \in \Sat_\infty(\Sat(\Ss_+))$. Finally, that the
second item in Proposition~\ref{prop:pumping-i} holds (\emph{i.e.}, that any
language in $\Kb''$ can be defined by a \fo formula of rank at most
$|B| \cdot 2^{|\unclos{\Ss_+}|^2}+1$) comes from the following fact
(which generalizes Fact~\ref{fct:foconcat} to \iwords).

\begin{fact} \label{fct:foconcati}
Set $k \geq 0$. Let $L_1$ be a language and $L_2$ be an \ilang, each
defined by a first-order formula of rank at most $k$. Then
$L_1L_2$ can be defined by a first-order formula of rank at most
$k+1$.
\end{fact}

\section{Conclusion}

We have given simple and self-contained proofs that one can decide in
\textsc{Exptime} whether two regular languages of finite or infinite words are
separable by first-order logic. Further, we have obtained an upper bound on the
quantifier rank of an expected separator. We have also described a procedure to
compute, given as input a morphism $\alpha$ into a finite semigroup, a finite
set of \fo-formulas whose associated languages form a partition of $A^+$, and
such that any two \fo-separable languages recognized by $\alpha$ can be
separated by a disjunction of some of these formulas. These formulas are computed
inductively along the correctness proof of our algorithm.

There are some open questions left in this line of research. First, we do not
know if the bounds are tight. We conjecture that the problem is
\textsc{Exptime}-complete starting from semigroups. A related question is the
complexity, starting from NFAs. Our results imply a 2-\textsc{Exptime} upper
bound (for DFAs, checking first-order definability
is~\textsc{Pspace}-complete~\cite{DFA-sf-PSPACE}). Moreover, we do not know
whether the bounds on the quantifier depth and the size of the expected
separator are tight. Finally, it is likely that these techniques can be
extended to other settings without much difficulty, as for finite or infinite
Mazurkiewicz traces. A much more interesting and challenging problem is to
look at separation for tree languages, where, for first-order logic, even
getting a decidable characterization is open despite many recent attempts.

\bibliographystyle{abbrv}

\end{document}